\documentclass[seceq]{ptptex}

\usepackage{amsmath,amsfonts,latexsym,amssymb} 
\usepackage{verbatim,amsthm,graphicx} 
\usepackage{wrapft}

\usepackage{amsmath}
\usepackage{amssymb}
\usepackage{mathrsfs}

\def\defscript{\mathscr}
\def\A{{\defscript A}}

\def\D{{\defscript D}}
\def\E{{\defscript E}}
\def\F{{\defscript F}}

\def\K{{\defscript K}}
\def\LL{{\defscript L}}
\def\M{{\defscript M}}
\def\N{{\defscript N}}

\def\R{{\defscript R}}

\def\ifempty#1{\def\tmpdata{#1}\ifx\tmpdata\empty }

\def\linebreak{\hfill\break}

\def\bra<#1|{\langle #1\rvert}
\def\ket|#1>{\lvert#1 \rangle}
\def\braket<#1|#2>{\langle #1|#2 \rangle}

\def\pfrac#1#2{\left(\frac{#1}{#2}\right)}

\def\otop#1{\hbox{$#1\kern-0.1em$\llap{\hbox{\raise1.7ex\hbox{$\scriptstyle\circ$}}}} }

\def\inpare#1{\left(#1\right)}
\def\bigpare(#1){\left(#1\right)}
\def\inrbra#1{\left\{ #1 \right\}}
\def\insbra#1{\left[ #1 \right]}

\def\bigbra[#1]{\left[ #1 \right]}

\def\h{\hat }

\def\d{\dot }

\def\b{\bar }

\def\tend{\rightarrow}

\def\therefore{\mbox{\setbox0=\hbox{X}\hbox{$\ldotp$}\raise0.7\ht0\hbox{$\ldotp$}\hbox{$\ldotp$}} \quad }
\def\because{\mbox{\setbox0=\hbox{X}\raise0.7\ht0\hbox{$\ldotp$}\hbox{$\ldotp$}\raise0.7\ht0\hbox{$\ldotp$}}\kern0pt }

\def\bm#1{\boldsymbol{#1}}

\def\RF{{{\mathbb R}}}

\def\Set#1{\left\{#1\right\}}

\def\upin{\hbox{\setbox0=\hbox{$\cup$} \vrule width 0.05 \wd0 height \ht0 depth 0pt \kern - 0.5\wd0 \box0 }}
\def\Frac(#1/#2){\left(\frac{#1}{#2}\right)}

\def\sdprod{\mathrel{{\setbox0=\hbox{$\displaystyle\times$}\lower0.3\wd0\hbox{$\stackrel{\box0}{\scriptstyle\sim}$}}}}

\def\w{\wedge}

\def\tosigma#1,{%
    \ifx\tmpindex\relax \def\tmpindex{#1} \let\next=\tosigma
    \else \ifnum\tmpindex=0 1 \else \sigma_\tmpindex \fi
          \ifx#1\relax  \let\next=\relax
          \else \otimes \let\next=\tosigma \def\tmpindex{#1} \fi
    \fi \next}
\def\tspb(#1){\let\tmpindex=\relax\tosigma#1,\relax,}
\def\Order#1{{\rm O}\!\left(#1\right)}

\def\HyperG(#1,#2;#3;#4){F\inpare{\textstyle #1,#2;#3;#4}}

\def\Lie{\hbox{\rlap{$\cal L$}$-$}}

\def\THB{{\mathbb T}}
\def\VHB{{\mathbb V}}
\def\SHB{{\mathbb S}}

\def\Eq#1{\begin{equation} #1 \end{equation}}
\def\Eqn#1{\Eq{#1 \nonumber}}
\def\Eqr#1{\begin{eqnarray} #1 \end{eqnarray}}

\def\Eqrsub#1{\begin{subequations}\Eqr{#1}\end{subequations}}
\def\Eqrsubl#1#2{\begin{subequations}
  \expandafter\ifx\csname Rlabel\endcsname \relax \label{#1}
  \else \Rlabel{#1} \fi \Eqr{#2}\end{subequations}}
\def\Bitm{\begin{itemize}}
\def\Eitm{\end{itemize}}
\def\Blist#1#2{\begin{list}{#1}{\parsep=0pt \itemsep=0pt%
  \listparindent=0pt #2}}
\def\Elist{\end{list}}
\long\def\ignore#1#2{\def\ignoreflag{#1}\long\def\tmptext{#2}
  \ifnum\ignoreflag>1 #2 \fi}

\usepackage{amsthm}

\theoremstyle{definition}
\newtheorem{definition}{Definition}[section]

\newtheorem{theorem}[definition]{Theorem}

\def\THB{{\mathbb T}}
\def\VHB{{\mathbb V}}
\def\SHB{{\mathbb S}}

\def\FigDir{.}

\title{
Perturbations and Stability of Static Black Holes in Higher Dimensions%
}

\author{
Akihiro \textsc{Ishibashi}$^{a,c}$ %
and 
Hideo \textsc{Kodama}$^{a,b}$ %
}

\inst{%
$^{a}$Theory Center, IPNS, KEK, 1-1 Oho, Tsukuba 305-0801, Japan 
\\ 
$^{b}$The Graduate University of Advanced Studies, 
1-1 Oho, Tsukuba 305-0801, Japan 
\\ 
$^{c}$Department of Physics, Kinki University, Higashi-Osaka 577-8502, Japan
}

\abst{%
In this chapter we consider perturbations and stability of 
higher dimensional black holes focusing on the static background case. 
We first review a gauge-invariant formalism for linear perturbations in  
a fairly generic class of $(m+n)$-dimensional spacetimes with a 
warped product metric, including black hole geometry. 
We classify perturbations of such a background into three types, 
the tensor, vector and scalar-type, 
according to their tensorial behavior on the $n$-dimensional part of 
the background spacetime, and for each type of perturbations, we introduce 
a set of manifestly gauge invariant variables. We then introduce 
harmonic tensors and write down the equations of motion for the expansion 
coefficients of the gauge invariant perturbation variables in terms of the 
harmonics. In particular, 
for the tensor-type perturbations a single master equation 
is obtained in the $(m+n)$-dimensional background, which is applicable 
for perturbation analysis of not only static black holes but also 
some class of rotating black holes as well as black-branes. 
For the vector and scalar type, 
we derive a set of decoupled master equations when the background is 
a $(2+n)$-dimensional static black hole in the Einstein-Maxwell theory with 
a cosmological constant.   
As an application of the master equations, we review the stability analysis of higher 
dimensional charged static black holes with a cosmological constant. 
We also briefly review the recent results of a generalization 
of the perturbation formulae presented here and stability analysis 
to static black holes in generic Lovelock theory. 
}

\begin{document}

\maketitle

\tableofcontents 

\section{Introduction}  

There is a large variety of black hole solutions in higher dimensions. 
The stability of such exact solutions is clearly an important issue,
and analysing linear perturbations of the existing exact solutions would
be the first step to take. 
If a stationary black hole solution is shown to be stable under perturbations, 
it implies that the solution describes a possible final state of 
dynamical evolution of a gravitating system. If, on the other hand, 
an instability is found, it then indicates the existence of a different 
branch of solutions, which the original 
solution may decay into, and one can then anticipate more variety 
of black hole solutions. 

Apart from the stability issue, perturbation analysis also tells us 
a lot about basic properties of black hole solutions. 
The spectra of quasinormal modes\cite{Berti.E&Cardoso&Starinets2009} contain 
information about the geometric structure of the background metric, 
especially near the horizon.  
The study of stationary perturbations of a stationary black hole solution 
provides a criterion for the uniqueness/non-uniqueness property of 
the solution. Considering such stationary perturbations of a known 
solution may also be useful when attempting to construct approximate
solutions.

In this article, we shall review the perturbation formulae for higher 
dimensional static black holes and the stability analysis, following,
to a large extent,
Refs.~\citen{Kodama.H&Ishibashi2003A,Ishibashi.A&Kodama2003A,Kodama.H&Ishibashi2004A}. 
The linearised Einstein equations off of a black hole spacetime are in general 
still quite involved, having a number of perturbation variables intricately coupled, 
and one therefore needs to simplify them to a tractable form. 
In $4$-dimensions, the perturbed Einstein equations for a stationary 
vacuum black hole solution can be reduced to a set of simple decoupled 
ordinary differential equations (ODEs) by exploiting some particular 
geometric feature of the background solution: the Regge-Wheeler-Zerilli 
equations\cite{ReggeWheeler1957,Zerilli.F1974} for the Schwarzschild metric 
and Teukolsky equations\cite{Teukolsky72} for the Kerr metric.  
For higher dimensional rotating black holes, we are still a long way 
from having such a complete formulation for perturbations, 
though considerable progress along this direction has recently been made 
for some special cases\cite{KLR06,Murata&Soda08CQG,Murata&Soda08PTP,DR09} 
[see Chapter~7 and references therein for the rotating black hole case].  
Fortunately, for static black holes, e.g., Schwarzschild-Tangherlini metric 
and its cousins, such a reduction is possible in arbitrary higher 
dimensions, and a set of decoupled master equations, which correspond to 
the Regge-Wheeler-Zerilli equations in $4$-dimensions, are now 
available\cite{Kodama.H&Ishibashi2003A,Kodama.H&Ishibashi2004A}. 
More precisely, consider a $(2+n)$-dimensional static black hole 
in Einstein-Maxwell system with a cosmological constant as our background, 
in which the $n$-dimensional internal space is an Einstein space 
and describes the horizon cross-section geometry. 
Then perturbation variables are classified 
into three types according to their tensorial behavior on 
the internal space. For each type of perturbations, a basis of gauge-invariant 
variables is introduced, and the Einstein and Maxwell equations are written 
in terms of them. For each type of perturbations, the perturbed equations 
of motion are further reduced to a set of decoupled master equations 
for a single scalar variable on the $2$-dimensional part of the background 
spacetime. Furthermore, the master equations thus obtained in 
Refs.~\citen{Kodama.H&Ishibashi2003A,Kodama.H&Ishibashi2004A} 
take the form of a second-order self-adjoint ODE with respect to the radial 
coordinate of the black hole background and are 
therefore immediately applied to the stability 
analysis\cite{Ishibashi.A&Kodama2003A,Kodama.H&Ishibashi2004A}.

In the next section we first establish our notation and convention 
by describing our background geometry. We next explain how to decompose 
tensor fields in our background spacetime and how to construct manifestly 
gauge-invariant perturbation variables. 
In section~\ref{sec:Harmonictensors}, we introduce harmonic tensors 
on $n$-dimensional Einstein space and provide several theorems concerning 
basic properties of the harmonic tensors.
We also give explicit expressions for scalar, vector and tensor harmonics on $S^n$ in terms of the homogeneous coordinates for $S^n$ in $E^{n+1}$.  
Subsequently from section~\ref{sec:TensorPerturbation} to
section~\ref{sec:ScalarPerturbation}, we focus on perturbations of static 
black holes in the Einstein-Maxwell system with cosmological constant and 
derive a set of master equations. In section~\ref{sec:Lovelock}, 
we describe how to derive the master equations for static black holes 
in most general Lovelock theory recently obtained\cite{Takahashi.T&Soda2010A}. 
Then in section~\ref{sec:stabilityanalysis}, using the master equations, 
we show that a large class of static black holes in the Einstein-Maxwell 
system with cosmological constant are stable with respect to linear 
gravitational as well as electromagnetic perturbations. 
We also briefly comment on the recent study of instability of Lovelock 
black holes. 
Section~\ref{sec:SummaryDiscussions} is devoted to summary and discussion.   
 
\section{Notation and conventions}  
\label{sec:background:decomposition:harmonics}

In this section we describe our background spacetime and discuss 
how to classify tensor fields on the spacetime. We then discuss the
problem of gauge freedom and introduce gauge-invariant variables.
We generally follow the notation and conventions of 
the papers~\citen{Kodama.H&Ishibashi2003A}, \citen{Ishibashi.A&Kodama2003A}, 
and \citen{Kodama.H&Ishibashi2004A}, which throughout this chapter we
refer to as Papers~I, II, and III, respectively.

\subsection{Background Geometry}
We consider an $(m+n)$-dimensional spacetime whose manifold structure 
is locally a warped product type, ${\M}= {\N}^m \times {\K}^n$, 
and accordingly we often need to distinguish between tensors living in these 
different manifolds, ${\M}$, ${\N}^m$, and ${\K}^n$. 
For this reason, we do not employ the abstract index notation\cite{Wald84} 
in this chapter, and 
instead, we use upper case latin indices in the range $K,L,M,N,\dots$ to 
denote tensors on $\M$, 
lower case latin indices in the range $a,b,\dots,h$ on $\N^m$, and 
lower case latin indices in the range $i,j,\dots, p$ on $\K^n$. 
Accordingly, we introduce coordinates 
$x^M = (y^a,z^i)$ in terms of which our background metric is written 
\Eq{
 g_{MN}dx^M dx^N 
 = g_{ab}(y) dy^a dy^ b+ r^2(y)\gamma(z)_{ij}dz^i dz^j \,. 
\label{BG:metric}
}
We assume that an $m$-dimensional metric $g_{ab}(y)$ on $\N^m$ 
is Lorentzian and an $n$-dimensional internal metric $\gamma(z)_{ij}$ on $\K^n$ is 
{\em Einstein} i.e., 
\Eq{{\hat R}_{ij}= (n-1)K\gamma_{ij} 
 \label{EinsteinSpace}
}
for some constant $K$, with ${\hat R}_{ij}$ being 
the Ricci tensor of $\gamma_{ij}$. 
When $\K^n$ is maximally symmetric, the constant $K$ corresponds 
to the sectional curvature of $\K^n$, and in what follows 
we normalize $K=0,\pm 1$. We assume that $\K^n$ be complete, 
as it describes the geometry of cross-sections of the event horizon.

We denote by $\nabla_M$, $D_a$, and $\hat D_i$, the covariant 
derivative operators compatible with $g_{MN}$, $g_{ab}$, and 
$\gamma_{ij}$, respectively.  
Having these derivative operators, we can define the curvature tensors 
on $\M$, $\N^m$ and $\K^n$, and find their relations in terms of 
the coordinate components as 
\Eqrsub{
&& R^a{}_{bcd}={}^m\! R^a{}_{bcd} \,,\\
&& R^a{}_{ibj}=-\frac{D^aD_b r}{r}g_{ij} \,,\\
&& R^i{}_{jkl}=\hat R^i{}_{jkl}
         -(Dr)^2(\delta^i_k\gamma_{jl}-\delta^i_l\gamma_{jk}) \,,
}
where ${}^m\! R^a{}_{bcd}$ and $\hat R^i{}_{jkl}$ are 
the curvature tensors of $g_{ab}$ and $\gamma_{ij}$, respectively. 
{}From this and Eq.~\eqref{EinsteinSpace}, we obtain
\Eqrsubl{BG:Ricci}{
& R_{ab}&=\frac{1}{2}{}^m\! R g_{ab} -n\frac{D_aD_b r}{r} \,,\\
& R_{ai}&=0 \,,\\
& R_{ij}&= \left(- \frac{\Box r}{r}
        +(n-1)\frac{K-(Dr)^2}{r^2} \right) g_{ij} \,,\\
& R &={}^m\! R
    -2n\frac{\Box r}{r}+n(n-1)\frac{K-(Dr)^2}{r^2} \,, 
}
where $\Box=D^aD_a$ and ${}^m\! R$ is the scalar curvature of $g_{ab}$. 
Note that the Ricci tensor takes the same form as in the case 
in which $\K^n$ is maximally symmetric\cite{Birmingham.D1999}. 

The geometric structure of our background spacetime requires that 
the background stress-energy tensor $T_{MN}$ should take the form 
\Eq{
 T_{ai}=0 \,, \quad T^i{}_j = P \delta^i{}_j \,,  
\label{BG:energymomentum} 
} 
where $P$ is a scalar field on $\N^m$.  

Then, the background Einstein equations, including cosmological 
constant, $\Lambda$, are written as 
\Eqr{
 && 
 {}^m\!G_{ab}-n\frac{D_aD_br}{r}
    -\left(
           \frac{n(n-1)}{2}\frac{K-(Dr)^2}{r^2} -n\frac{\Box r}{r}
     \right)g_{ab}
 = -\Lambda g_{ab} + \kappa^2 T_{ab} \,, 
\label{BG:Einstein:ab}
\\ 
 && -\frac{1}{2}{}^m\!R 
             -\frac{(n-1)(n-2)}{2}\frac{K-(Dr)^2}{r^2} 
             + (n-1)\frac{\Box r}{r} 
 =  \kappa^2 P -\Lambda \,,   
\label{BG:Einstein:ij} 
} 
where $\kappa^2$ denotes the gravitational constant. 
Note that although our main concern is about static black holes for 
which $m=2$ with the metric form given by Eq.~\eqref{metric:GSBH} below, 
our metric ansatz above also allows us to consider various different
geometries, such as black-string/branes when $m \geq 3$ and 
Myers-Perry black holes with a single rotation when $m=4$.

\subsection{Static black holes in $(2+n)$-dimensions} 
Now, having static black hole geometry in mind, let us set $m=2$ and 
consider an electromagnetic field $\F_{MN}$   
as a source for the background gravitational field. 
The field strength $\F_{MN}$ may be given by 
\Eq{
\F=\frac{1}{2}E_0 \epsilon_{ab}dy^a\wedge dy^b 
    +\frac{1}{2}\F_{ij}dz^i\wedge dz^j\,. 
\label{BG:EMfield}
}
Then, from $\nabla_{[M}\F_{N L]}=0$, we obtain 
\Eq{
   E_0=E_0(y) \,,\quad
  \F_{ij}=\F_{ij}(z) \,,\quad 
  \partial_{[k}\F_{ij]}=0 \,,
}
and from $\nabla_N \F^{MN}=0$, 
\Eqrsub{
&& 0=\nabla_N \F^{aN}=\frac{1}{r^n}\epsilon^{ab}D_b(r^nE_0)\,, 
\\
&& 0=\nabla_N \F^{iN}=\hat D_j\F^{ij} \,. 
}
These equations imply that the electric field $E_0$ takes the 
Coulomb form, 
\Eq{
E_0=\frac{q}{r^n} \,,  
\label{BG:Efield}
}
and $\hat\F=\frac{1}{2}\F_{ij}(z)dz^i\wedge dz^j$ is a harmonic form 
on $\K^n$. Although, in general, there may exist such a harmonic form that 
produces an energy-momentum tensor consistent with the structure of 
the Ricci tensors in Eq.~\eqref{BG:Ricci}, in the following we consider 
only the case $\F_{ij}=0$. 
With this assumption, the energy-momentum tensor for the electromagnetic 
field,  
\Eq{
 T^{{\rm (em)}}_{\ MN}=\F_{M L}\F_\nu{}^L
         -\frac{1}{4}g_{MN}\F_{LK}\F^{LK} \,,
}
is written 
\Eq{
T^{{\rm (em)}}{}^a{}_b=-P \delta^a{}_b \,, \quad 
T^{{\rm (em)}\: i}{}_j=P\delta^i{}_j \,; \quad
P=\frac{1}{2}E_0^2=\frac{q^2}{2r^{2n}} \,. 
\label{EMtensor:BG}
}
Then, solving the background Einstein equations, we have,  
when $\nabla r\not=0$, the black hole type solution 
\Eq{
ds^2=-f(r)dt^2+\frac{dr^2}{f(r)}+r^2d\sigma_n^2 \,,
\label{metric:GSBH}
}
where   
\Eq{
f(r)=K-\lambda r^2 -\frac{2M}{r^{n-1}}+\frac{Q^2}{r^{2n-2}} \,, 
\label{f:RNBH}
}
and
\Eq{
\lambda:=\frac{2\Lambda}{n(n+1)} \,,\quad
Q^2:=\frac{\kappa^2 q^2}{n(n-1)} \,.  
}
The spacetime described by this metric can contain a regular black hole 
for some restricted ranges of the parameters, $M$, $Q$, $\lambda$, $K$.
[For the allowed parameter regions for regular black holes, 
see Appendix~A of Paper III.]

Note that we can also consider, as our background, solutions with 
$\nabla r =0$, for which $\M$ becomes a cartesian product of a 
$2$-dimensional maximally symmetric spacetime $\N^2$ and 
the Einstein space $\K^n$. 
Such solutions include the Nariai solution\cite{Nariai.H1950,Nariai.H1961}, 
as well as the Bertotti-Robinson solution in $4$-dimensions. 
The solutions, and their parameter ranges are given 
in Paper III 
and also in Ref.~\citen{Cardoso.V&Dias&Lemos2004}. 
As in Paper~III, we can study perturbations of such Nariai-type
solutions in a similar manner as in the black hole background case, but
in this chapter, we are not going to deal with this case.

\subsection{Decomposition of vectors and symmetric tensors on compact manifolds}

When considering perturbations, we will in general have to deal with 
two major issues. One is concerning the ambiguity in a choice of gauge, 
due to the invariance of the Einstein equations under an infinitesimal 
gauge transformation. 
This issue will be discussed in the next subsection. 
The other one is that even when linearised, the Einstein equations 
are still a set of intricately coupled equations for a number of
perturbation variables of the metric and matter fields, and
are in general very difficult to solve. We therefore need to reduce
the linearised Einstein equations to a simple tractable form. 
For this purpose we first classify perturbation variables into three different 
types according to their tensorial behavior on $\K^n$ in such 
a way that the linearised Einstein equations get decoupled and can be dealt
with separately for each type of perturbations.      
We then introduce harmonic tensors on ${\K}^n$ so that 
each type of the perturbed Einstein equations reduces 
to a set of equations on the $m$-dimensional spacetime $(\N^m,g_{ab})$.  
This procedure enables us to obtain a set of significantly simplified
equations for master scalar variables as we will see in later sections.

Let us consider how to classify perturbation variables. We first note 
the following two decomposition theorems:

\smallskip 
\noindent 
(i) Suppose $({\K}^n, \gamma_{ij})$ be a compact Riemannian manifold. 
Any dual vector field on ${\K}^n$ can be uniquely decomposed as 
\Eqr{
 v_i=V_i+{\hat D}_i S 
} 
where ${\hat D}^iV_i=0$. 
This is essentially the well-known Hodge decomposition theorem, 
and we refer to $V_i$ and $S$, respectively, as the vector- and 
scalar-type components of the dual vector $v_i$. 

\smallskip 
\noindent 
(ii) Suppose $({\K}^n, \gamma_{ij})$ be a compact Riemannian Einstein 
space, ${\hat R}_{ij} = c \gamma_{ij}$ for some constant $c$. 
Any second rank symmetric tensor field $t_{ij}$ can be uniquely 
decomposed as 
\Eqr{ 
t_{ij} &=& t^{(2)}_{ij} + 2{\hat D}_{(i} t^{(1)}{}_{j)} 
           + t_L \gamma_{ij} + \hat{L}_{ij}t_T 
\,, 
\\ 
    {\hat L}_{ij} &:=& 
    {\hat D}_i {\hat D}_j -\frac{1}{n}\gamma_{ij}{\hat \triangle} 
\label{def:Lij}
}
where ${\hat D}^i t^{(2)}_{ij} =0$, $t^{(2)}{}^i{}_i=0$, 
${\hat D}^i t^{(1)}{}_{i}=0$, and $t_L= t^m{}_m/n$. 
We refer to $t^{(2)}_{ij}$, $t^{(1)}_{i}$, and $(t_T, t_L)$, respectively, 
as the tensor-, vector- and scalar-type components of $t_{ij}$. 
Note that tensor component $t^{(2)}_{ij}$ exists only when $n \geq 3$.

A similar decomposition theorem--in which $\K^n$ is considered to be
maximally symmetric
--has been proved by Kodama and Sasaki\cite{Kodama.H&Sasaki1984}. 
%
%
A general proof of the above theorems (i) and (ii) is found
in Ref.~\citen{Ishibashi.A&Wald2004}, in which the compactness of ${\K}^n$ 
is used in an essential way to show that any symmetric elliptic operator, 
such as ${\hat D}^i{\hat D}_{(i}V_{j)}$ on $V_i$, 
is essentially self-adjoint and the spectrum 
of any self-adjoint operator is discrete. 
The symmetric, elliptic operators appeared in the proof will be 
essentially self-adjoint on many non-compact manifolds 
of interests, even though their spectra will be continuous, 
and analogous decomposition theorems should also hold for many non-compact 
manifolds. 
In the following when considering the case in which ${\K}^n$ is 
non-compact, we simply assume that analogous decomposition results hold.

Now let us consider metric perturbations $h_{MN} = \delta g_{MN}$ 
on our background spacetime $(\M,g_{MN})$. We may project 
$h_{MN}$ relative to the Einstein manifold ${\K}^n$ as 
\Eqr{
 h_{MN}dx^M dx^N = h_{ab}dy^ady^b + 2h_{ai}dy^adz^i+h_{ij}dz^idz^j \,.
}
The component $h_{ab}$ is purely scalar with respect to transformations 
on ${\K}^n$. As for the components $h_{ai}$ and $h_{ij}$, 
applying the above decomposition theorems (i) (ii), we further decompose them 
into their scalar, vector and tensor parts with respect to ${\K}^n$ as 
\Eqr{
h_{ai} &=& \hat{D}_ih_a + h^{(1)}_{ai} \,, 
\\
h_{ij} &=& h_T^{(2)}{}_{ij} + 2{\hat D}_{(i} h_T^{(1)}{}_{j)} + h_L \gamma_{ij}
            + {\hat L}_{ij} h_T^{(0)} \,,
}
where    
\Eqr{
  {\hat D}^jh_T^{(2)}{}_{ij} &=& h_T^{(2)}{}^i{}_{i} = 0 \,, 
\\
  {\hat D}^ah^{(1)}_{ai} &=& 0 \,, \quad 
      {\hat D}^ih_T^{(1)}{}_{i}  = 0 \,. 
} 
Thus, the tensor part of $h_{MN}$ is $h_T^{(2)}{}_{ij}$, 
the vector part of $h_{MN}$ consists of $(h^{(1)}_{ai}, h_T^{(1)}{}_i)$, 
and the scalar part of $h_{MN}$ consists of $(h_{ab},h_a,h_L,h_T^{(0)})$.

Similarly, we can decompose perturbations of the energy-momentum tensor,  
$\delta  T_{MN}$, into the tensor part $\delta T^{(2)}_{Tij}$,
the vector part $(\delta T^{(1)}_{ai}, \:\delta T^{(1)}_{Ti})$, and 
the scalar part $(\delta T_{ab}, \: \delta T_{a}, \: \delta T_L, 
\: \delta T_T)$, where   
\Eqr{
\delta T_{ai} &=& {\hat D}_i\delta T_{a}+ \delta T^{(1)}_{ai} \,,
\\
\delta T_{ij} &=& \delta T^{(2)}_{Tij} + 2{\hat D}_{(i} \delta T_T^{(1)}{}_{j)}
                + \delta T_L \gamma_{ij} 
                + {\hat L}_{ij} \delta T_T \,,
} 
with ${\hat D}^i\delta T_T^{(2)}{}_{ij}=0,\:\delta T_T^{(2)}{}^i{}_i=0$, 
${\hat D}^i\delta T^{(1)}_{ai}=0 = {\hat D}^i\delta T_T^{(1)}{}_{i}$.

The linearised Einstein equations are decomposed into these three
types, 
and thus we can deal with perturbations of each type separately.
Before writing down the equations for each type of perturbations, 
we discuss the gauge issue. 

\subsection{Gauge-invariant formulation}  

As is well-known, the Einstein equations are invariant  
under gauge transformation generated by an (infinitesimal)  
vector field $\xi^M $. Accordingly, a perturbation variable and 
its gauge transformed one, such as $\delta g_{MN}$ 
and $\delta g_{MN}- \Lie_\xi g_{MN}$, both must describe the same 
physical situation, hence giving rise to an ambiguity in the representation 
of perturbation variables. In order to remove this gauge ambiguity and 
extract the physical degrees of freedom, we may proceed either by 
imposing appropriate conditions that completely fix the gauge freedom, 
or by constructing manifestly gauge-invariant variables and 
writing down the relevant equations in terms of them. 
The two approaches are equivalent. To see this, let us have a look 
at gauge transformation law of perturbation variables.

In terms of the coordinates $(y^a,z^i)$, the gauge transformation 
of the metric perturbation $h_{MN}= \delta g_{MN}$ is written as 
\Eqr{
 h_{ab} &\rightarrow& h_{ab} - D_a \xi_b - D_b \xi_a \,,  
\\
 h_{ai} &\rightarrow& h_{ai} - r^2 D_a\left(\frac{\xi_i}{r^2} \right) 
                             - {\hat D}_i \xi_a \,, 
\\ 
 h_{ij} &\rightarrow& h_{ij} - 2{\hat D}_{(i} \xi_{j)} 
                             - 2r (D^a r) \xi_a \gamma_{ij} \,. 
}
Similarly, the gauge transformation of perturbation of 
the energy-momentum tensor is given by 
\Eqr{ 
 \delta T_{ab} &\rightarrow& \delta T_{ab} -\xi^cD_cT_{ab}
                                           -T_{ac}D_b\xi^c
                                           -T_{bc}D_a\xi^c \,, 
\label{pert:Tab} 
\\ 
 \delta T_{a i} &\rightarrow& \delta T_{a i} - T_{ab}{\hat D}_i\xi^b 
                                             - r^2 P D_a(\xi_i/r^2) \,,
\label{pert:Tai}
\\ 
 \delta T_{ij} &\rightarrow& \delta T_{ij} - \xi^a D_a(r^2P)\gamma_{ij} 
                                           -P(  
                   {\hat D}_i \xi_j + {\hat D}_j \xi_i  
                                             ) \,.
\label{pert:Tij}
}

We can of course classify the gauge transformations above
into the three tensorial types, and discuss each type separately. 
As is clear from the decomposition theorem~(i), the generator $\xi^M$
has only the vector- and scalar-type component; it is decomposed as 
\Eq{ 
    \xi_a = T_a \,, \quad \xi_i = V_{i}+ {\hat D}_i S \,, 
   } 
where $\gamma^{ij}{\hat D}_iV_j = 0$.    
Any tensor-type perturbation variable is by itself 
gauge-invariant: we have 
$(h_T^{(2)}{}_{ij}, \: \delta T_T^{(2)}{}_{ij})$. 
For the vector type perturbations, the gauge transformation law 
of the metric perturbation is given as follows: 
\Eqr{
&& 
 h^{(1)}_{ai} \rightarrow h^{(1)}_{ai}
                           -r^2 D_a \left(\frac{V_i}{r^2}\right)\,, 
\\  
 &&  h^{(1)}_{Ti} \rightarrow h^{(1)}_{Ti} - V_i\,.
 }
Inspecting the transformation laws, we find the following combination 
is gauge-invariant: 
\Eqr{
  F^{(1)}_{ai} = h^{(1)}_{ai}-r^2 D_a\left(\frac{h^{(1)}_{Ti}}{r^2}
  \right)
\label{def:F1ai}  
    }
Similarly, inspecting the gauge-transformation law of 
the matter perturbation, we find two gauge-invariant variables: 
\Eqr{
\tau^{(1)}_{ai} &:=& \delta T^{(1)}_{ai} -Ph^{(1)}_{ai} \,,
\label{def:tau1ai}  
\\
\tau^{(1)}_{ij} &:=& 2{\hat D}_{(i} \delta T_T^{(1)}{}_{j)}
                   -2P {\hat D}_{(i} h_T^{(1)}{}_{j)} \,.
\label{def:tau1ij}  
}
Any vector-type gauge invariant variable can be expressed as a 
linear combination of 
$(F^{(1)}_{ai}, \:\tau^{(1)}_{ai},\:\tau^{(1)}_{ij})$ 
and their derivatives.

The above equations, \eqref{def:F1ai}, \eqref{def:tau1ai}, \eqref{def:tau1ij}, 
defining the gauge invariant variables may be viewed in a way that 
$(h^{(1)}_{ai}, \delta T^{(1)}_{ai}, \delta T_T^{(1)}{}_{j})$ are expressed 
in terms of the gauge-invariants 
$(F^{(1)}_{ai}, \tau^{(1)}_{ai},\tau^{(1)}_{ij})$, and $h^{(1)}_{Ti}$. 
Since, under gauge-transformation, $h^{(1)}_{Ti}$ behaves just like 
$-\xi_i$, it can be chosen to take any value, and therefore one may view 
that $h^{(1)}_{Ti}$ alone is responsible for the gauge ambiguity. 
This, in turn, implies that the specification of $h^{(1)}_{Ti}$ 
in terms of $(F^{(1)}_{ai},\tau^{(1)}_{aj},\tau^{(1)}_{ij})$, corresponds to 
fixing the gauge freedom of the vector type perturbation.

For the scalar-type metric perturbation,
the gauge transformation law is explicitly given as follows: 
 \Eqr{
&& 
 h_{ab} \rightarrow h_{ab} - 2 D_{(a} T_{b)} \,, 
\\
&&
 h_{a} \rightarrow h_{a} -T_a - r^2 D_a \left(\frac{S}{r^2}\right) \,,  
\\
&&
 h_L \rightarrow h_L - 2 r({D^ar})T_a - \frac{2}{n}{\hat \triangle}S \,,  
\\
&&  
 h_T \rightarrow h_T - 2 S \,. 
}
Now let us define ${X_M = (X_a,X_i={\hat D}_i X_L)}$ by 
\Eq{
X_a := -h_a + \frac{r^2}{2}D_a\left(\frac{h_T}{r^2}\right) \,, 
\quad     
X_L := - \frac{h_T}{2}  \,. 
\label{def:X}
}
Then, noting that $X_M$ gauge-transforms as $X_M \rightarrow 
X_M + \xi_M$, i.e.,  
\Eq{
(X_a,X_L) \rightarrow (X_a+T_a, X_L + S) \,, 
} 
we can immediately find gauge-invariant combinations:  
\Eqr{
 F^{(0)}{}_{ab} &=&  h_{ab} + 2D_{(a}X_{b)} \,, 
\label{def:Fab}
\\
 F^{(0)} &=& h_L +2r (D^ar)X_a + \frac{2}{n}{\hat \triangle} X_L \,. 
\label{def:F}
}
As for the matter perturbation, we find gauge-invariant combinations: 
\Eqr{
\Sigma^{(0)}{}_{ab}
         &:=& \delta T_{ab} +X^cD_cT_{ab}+T_{ac}D_bX^c+T_{bc}D_aX^c \,,
\label{def:Sab}
\\
\Sigma^{(0)}{}_{ai} &:=& {\hat D}_i \delta T_{a} + T_{ab}{\hat D}_iX^b 
             + r^2 P D_a\left(\frac{{\hat D}_i X_L}{r^2} \right) \,,
\label{def:Sai}
\\
 \Sigma^{(0)} &:=& \delta T_L -Ph_L + r^2 X^aD_aP \,, 
\label{def:S}
\\
\Pi^{(0)}{}_{ij} &:= & {\hat L}_{ij}\delta T_T +2P {\hat L}_{ij} X_L\,.
\label{def:Piij}
}
Any scalar-type gauge invariant variable can be expressed as a 
linear combination of the gauge-invariant variables 
$(F^{(0)}{}_{ab},F^{(0)}{},\Sigma^{(0)}{}_{ab},\Sigma^{(0)}{}_{ai},
\Sigma^{(0)}{},\Pi^{(0)}{}_{ij})$ 
and their derivatives. 

One may view the above defining equations,  
\eqref{def:X}, \eqref{def:Fab},
\eqref{def:F}, \eqref{def:Sab}, \eqref{def:Sai}, \eqref{def:S},
\eqref{def:Piij}, in such a way that all the scalar-type perturbations 
$(h_{ab},h_a,h_L,h_T)$ and $(\delta T_{ab},\delta T_a,\delta T_L,\delta T_T)$
are expressed in terms of $F^{(0)}{}_{ab}$, $F^{(0)}$,
$\Sigma^{(0)}{}_{ab}$, $\Sigma^{(0)}{}_{ai}$, $\Sigma^{(0)}{}$,
$\Pi^{(0)}{}_{ij}$--which are gauge invariant, 
and $X_M$--which can be taken as completely arbitrary 
as the generator $\xi_M$ is. 
Then, one may say that specifying $X_M$ in terms of 
the above set of the gauge-invariant variables 
(or assigning $X_M$ some specific value) corresponds to 
fixing the gauge freedom. 
For example, the specification, $X_M =0$, corresponds 
to the longitudinal gauge often used in the cosmological context, 
when $m=1$, and to the Regge-Wheeler gauge, when $m=2$.  

The perturbed Einstein equations are gauge-invariant and therefore can 
be written in terms of the gauge-invariant variables introduced above.
Note that at this point all variables are functions on $\M$, being
dependent upon the $(m+n)$-coordinates, $x^M$.  
In the next section, we introduce tensor harmonics on $\K^n$,  
so that by expanding the perturbation variables in terms of the tensor 
harmonics, we can separate variables and reduce the relevant equations 
in $\M$ to equations in $\N^m$.

As a specific example of a source for gravitational field,
let us consider the Maxwell field in the $m=2$ case. 
Perturbation of the field strength $\delta \F_{MN}$ satisfies  
the perturbed Maxwell equations: 
\Eqr{
\nabla_{[M} \delta \F_{N L]}=0 \,, 
\label{MaxwellEq:Perturbation:dF}
\quad 
 \delta \left( \nabla_N \F^{MN} \right) =J^M \,,   
} 
where here and in the following the external current $J^M$ is treated 
as a first-order quantity. 
The first equation implies that 
$\delta\F_{M N}$ is expressed in terms of the perturbation of the vector 
potential, $\delta \A_M$, as 
$\delta\F_{MN} =2 \nabla_{[M}\delta\A_{N]}$. Therefore 
$\delta\F_{MN}$ does not contain any tensor-type perturbation. 
The second equation gives two set of equations, 
\Eqrsubl{MaxwellEq:Perturbation}{
&& \frac{1}{r^n}D_b(r^n (\delta\F)^{ab}) 
          +\hat D_i(\delta\F)^{ai}
          +E_0\epsilon^{ab}\left( 
          \frac{1}{2}D_b (h^i{}_i-h^c{}_c)-\hat D_i h^i{}_b\right) =J^a \,,
\label{MaxwellEq:Perturbation:Ja}\\
&& \frac{1}{r^{n-2}}D_a\left[ r^{n-2}
      \left( (\delta\F)_i{}^a+E_0\epsilon^{ab}h_{ib} \right) \right]
       +\hat D_j \delta\F_i{}^j =J_i \,,  
\label{MaxwellEq:Perturbation:Ji}
}
where note that 
$
(\delta \F)^{MN}=g^{ML}g^{N K}\delta \F_{LK}
$.

The contribution of electromagnetic field perturbations to 
the energy-momentum tensor are given by 
\Eqrsubl{EMFperturbation:EMtensor}{
& \delta T^{{\rm (em)}}_{ab}= & \frac{E_0}{2}\left(\epsilon^{cd}\delta\F_{cd}
         +E_0 h^c{}_c  \right)g_{ab}
         -\frac{1}{2}E_0^2 h_{ab} ,\\
& \delta T^{{\rm (em)}}{}^a{}_i= & - E_0\epsilon^{ab} \delta\F_{bi},\\
& \delta T^{{\rm (em)}}{}^i{}_j= &-\frac{E_0}{2}\left( \epsilon^{cd}\delta\F_{cd}
         +E_0 h^c_c  \right)\delta^i{}_j.        
}

For vector-type perturbations, we note that $\delta\F_{MN}$ is 
gauge-invariant as well as invariant under a coordinate gauge transformation 
$x^M \rightarrow x^M +\xi^M$. 
We find 
\Eqr{
  \delta \tau^{{{\rm (em)}}(1)}_{ai} = - E_0\epsilon_{ab}D^b\delta \A^{(1)}_i \,, 
\quad 
  \delta \tau^{{{\rm (em)}}(1)}_{ij} = 0 \,, 
}
where ${\hat D}^i \delta \A^{(1)}_i=0$. 
We find from the first of Eqs.~\eqref{MaxwellEq:Perturbation:Ji}  
that a vector perturbation of the Maxwell field can be expressed in 
terms of the gauge-invariant variable $\delta \A^{(1)}_i$ on $\N^2$ as, 
\Eq{ 
\delta \F_{ab}=0 \,, \ 
\delta \F_{ai}=D_a \delta \A^{(1)}_i \,, \
\delta \F_{ij}= \hat D_i \delta \A^{(1)}_j - \hat D_j \delta \A^{(1)}_i \,. 
\label{Pert:Vector:Maxwell:A}
}

\medskip

For scalar-type perturbations, $\delta \F_{MN}$ gauge-transforms as: 
\Eqr{
\delta \F_{ab} \rightarrow \delta \F_{ab} - D_c(E_0\xi^c)\epsilon_{ab} \,, 
\quad 
\delta \F_{aj} \rightarrow \delta \F_{aj}-E_0\epsilon_{ab}{\hat D}_j\xi^b\,.   
}
Note that $\delta \F_{ij} = 0 $ for a scalar perturbation. 
So, we can immediately find gauge-invariant combinations, 
${\E}^{(0)},\ {\E}^{(0)}_a$, given by 
\Eqr{
 \epsilon_{ab}{\E}^{(0)}&=& \delta \F_{ab}+ \epsilon_{ab} D_c(E_0X^c) \,,  
\label{def:gi:scalar:pert:Maxwell}
\\ 
 \epsilon_{ab} {\hat D}_i {\E}^{(0)b }
 &=& \delta \F_{ai} + E_0 \epsilon_{ab}{\hat D}_iX^b \,. 
\label{def:gi:scalar:pert:Maxwell:a}
} 
The gauge-invariant variables introduced above are then written as 
\Eqr{
\Sigma^{{\rm (em)}(0)}{}_{ab} &=& - \frac{E_0^2}{2}
                                \left(
                                      F^{(0)}{}_{ab}-F^{(0)}{}^c{}_cg_{ab}
                                \right)
                        - E_0 {\E}^{(0)}g_{ab} \,, 
\\ 
\Sigma^{{\rm (em)}(0)}{}_{ai} &=& - E_0 {\hat D}_i {\E}^{(0)}_a \,, 
\\
\Sigma^{{\rm (em)}(0)}{} &=& r^2
                      \left( 
                            E_0{\E}^{(0)} - \frac{E_0^2}{2}F^{(0)}{}^c{}_c 
                      \right) \,,
\\
\Pi^{{\rm (em)}(0)}{}_{ij} &=& 0 \,. 
}

\section{Harmonic tensors on the Einstein space $\K^n$} 
\label{sec:Harmonictensors}

When we write down perturbation equations and solve them, it is often more 
convenient to expand perturbation variables into Fourier-type harmonic 
components in terms of harmonic tensors appropriate for each tensorial type on 
the internal space $\K^n$. 
In this section, we summarize the definitions and basic properties of 
such harmonic tensors relevant to the descriptions in the present paper. 

\subsection{Scalar harmonics}

\subsubsection{Definition and properties}

Covariant derivative $\h D$ along $\K^n$ appears only in the form of 
the Laplace-Beltrami operator in the perturbation equations for 
scalar-type variables because $\gamma_{ij}$ is the only non-trivial symmetric 
tensor on the Einstein space $\K^n$. Hence, the scalar-type perturbation 
equations reduce to a set of PDEs on $\N$  by expanding scalar-type 
perturbation variables in terms of scalar harmonic functions on $\K^n$ that 
satisfy 
\Eq{
\left( \hat\triangle +k^2 \right)\SHB=0 \,.
}
Here, when $\K^n$ is non-compact, we assume that $-\hat\triangle$ is extended 
to a non-negative self-adjoint operator in the $L^2$-space of functions 
on $\K^n$. Hence, $k^2\ge0$. Such an extension is unique if $\K^n$ is 
complete\cite{Craioveanu.M&Puta&Rassias2001B} and is given by 
the Friedrichs self-adjoint extension of the symmetric and non-negative 
operator $-\hat\triangle$ on $C_0^\infty(\K^n)$.

If $\K^n$ is closed, the spectrum of $\hat\triangle$ is completely discrete, each eigenvalue has a finite multiplicity, and the lowest eigenvalue is $k^2$=0, whose eigenfunction is a constant. A perturbation corresponding to such a constant mode generally represents a variation of the parameters of the background solution such as $\lambda$, $M$, and $Q$. Thus, this mode is relevant to arguments on perturbative uniqueness of a background solution. Note that when $k^2=0$ is contained in the full spectrum but does not belong to the point spectrum, as in the case $\K^n=\RF^n$, it can be ignored without loss of generality.  

For modes with $k^2>0$, we can use the vector fields and the symmetric trace-free tensor fields defined by 
\Eqrsub{
&& \SHB_i=-\frac{1}{k} \hat D_i\SHB,\\
&& \SHB_{ij}=\frac{1}{k^2}\hat D_i\hat D_j\SHB
             +\frac{1}{n}\gamma_{ij}\SHB;\ \SHB^i{}_i=0
}
to expand vector and symmetric trace-free tensor fields, respectively. Note that $\SHB_i$ is also an eigenmode of the operator $\hat D\cdot\hat D$, i.e., 
\Eq{
[\hat D\cdot\hat D+k^2-(n-1)K]\SHB_i=0,
}
while $\SHB_{ij}$ satisfies
\Eq{
(\hat\triangle_L -k^2)\SHB_{ij}=0,
}
where $\h\triangle_L$ is the Lichnerowicz operator defined by
\Eq{
\hat\triangle_L h_{ij}:= -\h D\cdot \h D h_{ij} - 2 \h R_{ikjl} h^{kl} 
 + 2(n-1) K h_{ij}.
}
When $\K$ is a constant curvature space, this operator is related to the Laplace-Beltrami operator by
\Eq{
\h\triangle_L = -\h\triangle + 2nK,
}

In the case of scalar harmonics, the modes with $k^2=nK$ are 
exceptional. Given our assumption, these modes exist only for $K=1$. 
Because $\K^n$ is compact and closed in this case, from the identity
\Eq{
\hat D_j \SHB^j{}_i=\frac{n-1}{n}\frac{k^2-nK}{k}\SHB_i \,,
}
we have $\hat D_j\SHB^j{}_i=0$. From this, it follows that $\int d^nz 
\sqrt{\gamma}{\SHB_{ij}}^*\SHB^{ij}=0$. Hence, $\SHB_{ij}$ vanishes 
identically.  

We can further show that the second smallest eigenvalue for $-\h\triangle$ is equal to or greater than $nK$ when $K>0$. 
To see this, let us define $Q_{ij}$ by 
\Eqn{
Q_{ij}:= {\hat L}_{ij}Y 
       =\h D_i \h D_j Y -\frac{1}{n}\gamma_{ij}\h \triangle Y \,.
}
Then, we have the identity
\Eqn{
Q_{ij}Q^{ij}=\h D^i(\h D^i Y \h D_i\h D_j Y-Y\h D_i\h \triangle Y-\h R_{ij}\h D^jY)
  +Y\insbra{\h \triangle(\h \triangle +(n-1)K)}Y
  -\frac{1}{n}(\h \triangle Y)^2 \,.
}
For $Y=\SHB$, integrating this identity, we obtain the constraint on the second eigenvalue
\Eq{
k^2\ge nK \,.
}
For $\K^n=S^n$, the equality holds for the second smallest eigenvalue as we see soon.

\subsubsection{Harmonic functions on $S^n$}

Explicit expressions for the harmonic functions on higher-dimensional spaces 
sometimes become necessary to investigate the global structure of 
perturbations. The multiplicity of eigenvalues also has a crucial importance 
in applying the perturbation theory to black hole evaporation. 
Here, we give such information for harmonic functions on $S^n$.

There are several ways to express higher-dimensional spherical harmonic 
functions. For example, we can get expressions in terms of special functions 
by solving the recurrence relation obtained by dimensional reduction. 
In this paper, we give a different approach in which harmonic functions are 
expressed in terms of the homogeneous cartesian coordinates 
for the Euclidean space $E^{n+1}$ containing the unit sphere $S^n$. 

Let us denote the homogeneous cartesian coordinates of $S^n$ by $\Omega^A$ 
($A=1,\cdots,n+1$); $\Omega\cdot\Omega=1$. Then, we can show that 
$\Omega^A$ satisfies 
\Eqrsubl{HarmonicCoord:Properties}{
&& \h D_i \h D_j \Omega^A=-\gamma_{ij}\Omega^A,\\
&& \h \triangle \h D_i\Omega^A=-\h D_i \Omega^A, \\
&& \h D_i \Omega^A \h D^i\Omega^B=\delta^{AB}-\Omega^A\Omega^B.
}
From these formulae, the following theorem holds. 

\bigskip
\begin{theorem}[Scalar harmonics on $S^n$]\label{th:HarmonicScalar}
Let us define the function $Y_{\bm{a}}$ on $S^n$ by
\Eq{
Y_{\bm{a}}(\Omega)=a_{A_1\cdots A_\ell}\Omega^{A_1}\cdots \Omega^{A_\ell}
\label{HarmonicFunction:HCR}
}
in terms of a constant tensor $\bm{a}=(a_{A_1\cdots A_\ell})$ ($A_1,\cdots,A_\ell=1,\cdots,n+1$). Then, the following statements hold: 
\Bitm
\item[1)] $Y_{\bm{a}}$ is a harmonic function on $S^n$ with the eigenvalue
\Eq{
k^2=\ell(\ell+n-1),\quad \ell=0,1,2,\cdots,
}
if and only if $\bm{a}$ satisfies the conditions
\Eqrsubl{HarmonicScalar:HCR}{
&& a_{A_1\cdots A_\ell}=a_{(A_1\cdots A_\ell)},\\
&& a_{A_1\cdots A_{\ell-2}}{}^B{}_B=0\quad (\ell\ge2).
}
\item[2)] The harmonic functions $\Set{Y_{\bm{a}}}$ form a complete basis 
in $L^2(S^2)$. 
Two harmonic functions with different values of $\ell$ are orthogonal 
and those with the same $\ell$ have the inner product 
\Eqrsubl{HarmonicFunction:HCR:InnerProduct}{
&& (Y_{\bm{a}},Y_{\bm{a}'}):=\int d^n\Omega \bar Y_{\bm{a}}Y_{\bm{a}'}
= C(n,\ell) \b a^{j_1\cdots j_\ell}a'_{j_1\cdots j_\ell},\\
&& C(n,\ell)=\frac{2\pi^{\frac{n+1}{2}}\ell!}
               {2^\ell\Gamma\left(\frac{n+1}{2}+\ell\right)}.
}
\item[3)] The multiplicity of the $\ell$-eigenvalue is given by
\Eq{
N^\ell_S(S^n)=\frac{(n+2\ell-1)(n+\ell-2)!}{(n-1)!\,\ell!} \,. 
\label{HarmonicFunction:Multiplicity}
}
\Eitm
\end{theorem}

\begin{proof}
The first statement follows from  
\Eqr{
\h \triangle_n Y_{\bm{a}}
 &=&-\ell n Y_{\bm{a}}+\sum_{p\not=q} a_{A_1\cdots A_\ell} 
 \Omega^{A_1}\cdots \h D_k \Omega^{A_p}\cdots \h D^k \Omega^{A_q}\cdots \Omega^{A_\ell}
\nonumber\\
&=&-\ell (\ell +n-1)Y_{\bm{a}} + \frac{\ell(\ell-1)}{2}a_{A_1\cdots A_{\ell-2}}{}^B{}_B \Omega^{A_1}\cdots \Omega^{A_{\ell-2}}
}
which can be easily verified with the helps of the identities, 
Eqs.~\eqref{HarmonicCoord:Properties}. 

Next, all polynomials in the cartesian coordinates $x^A$ for $E^{n+1}$ are dense in the space of continuous functions in the unit cube in $E^{n+1}$, hence its restriction on $S^n$ is also dense in the space of continuous functions on $S^n$. This implies that all the harmonic functions of the type $Y_{\bm{a}}$ are dense in the function space $L^2(S^n)$. This proves the completeness. 
The inner product of these harmonic functions can be calculated as follows. First, by differentiating the function
\Eq{
F(r^2):=\int d^n\Omega e^{i \bm{r}\cdot\Omega}
       =\Omega_{n-1}\int_0^\pi d\theta\sin^{n-1}\theta e^{ir\cos\theta}
       =\Omega_n \frac{\Gamma\pfrac{n+1}{2}}{(r/2)^{(n-1)/2}} J_{\frac{n-1}{2}}(r)
}
repeatedly with respect to $x^A$ ($\bm{r}=(x^A)$), we obtain 
\Eq{
 \int d^n\Omega Y_{\bm{a}}(\Omega) e^{i\bm{r}\cdot\Omega}
 =(-1)^\ell  2^{\ell } x^{A_1}\cdots x^{A_{\ell }} a_{A_1\cdots A_{\ell }} F^{(\ell)}(r^2).
}
Differentiating this $\ell$-th times with respect to $x^A$ and putting $x^A=0$ yield Eq.~\eqref{HarmonicFunction:HCR:InnerProduct}. The multiplicity formula of the eigenvalue can be easily obtained by just counting the linearly independent solutions to Eq.~\eqref{HarmonicScalar:HCR}.
\end{proof}

\subsection{Vector harmonics}

\subsubsection{Definition and properties}

A harmonic vector is defined as a vector field on $S^n$ satisfying
\Eq{
   (\hat D\cdot\hat D +k_v^2)\VHB_i=0;\quad \hat{D}^i \VHB_i=0 \,. 
} 
From this we can define a symmetric trace-free tensor of rank $2$ by 
\Eq{
\VHB_{ij}=-\frac{1}{k_v}\hat D_{(i}\VHB_{j)} \,,
\label{VectorHarmonics:SymmetricTensor:def}
}
where the factor $1/k_v$ is just a convention (see below for the case 
$k_v=0$). This tensor is an eigentensor of the Lichnerowicz operator, 
\Eq{
\hat\triangle_L \VHB_{ij}=\left[ k_v^2+(n-1)K \right]\VHB_{ij} \,,
}
but is not an eigentensor of the Laplacian in general when $\K^n$ is not 
a constant curvature space.

In this paper, we assume that the Laplacian $-\hat D\cdot\hat D$ is extended 
to a non-negative self-adjoint operator in the $L^2$-space of 
divergence-free vector fields on $\K^n$, in order to guarantee the 
completeness of the vector harmonics. Because $-\hat D\cdot\hat D$ is 
symmetric and non-negative in the space consisting of smooth divergence-free 
vector fields with compact support, it always possesses a Friedrichs 
extension that has the desired property\cite{Akhiezer.N&Glazman1966B}. 
With this assumption, $k_v^2$ is non-negative.  

One subtlety that arises in this harmonic expansion concerns the zero modes of the Laplacian. If $\K^n$ is closed, from the integration of the identity $\hat D^i(V^j\hat D_i V_j)-\hat D^i V^j 
\hat D_i V_j=V^j\hat D\cdot\hat D V_j$, it follows that $\hat D_i \VHB_j=0$ for $k_v^2=0$. Hence, we cannot construct a harmonic tensor from such a vector harmonic. We obtain the same result even in the case in which $\K^n$ is open if we require that $\VHB^j\hat D_i \VHB_j$ fall off sufficiently rapidly at infinity. In the present paper, we assume that this fall-off condition is satisfied.  From the identity $\hat D^j\hat D_i V_j=\hat D_i\hat D^jV_j+(n-1)K\hat V_i$, such a zero mode exists only in the case $K=0$. We can further show that vector fields satisfying $\hat D_i V_j=0$ exist if and only if $\K^n$ is a product of a locally flat space and an Einstein manifold with vanishing Ricci tensor.

More generally, $\VHB_{ij}$ vanishes if $\VHB_i$ is a Killing vector. In this case, from the relation
\Eq{
2k_v \hat D_j \VHB^j{}_i= \left[ k_v^2 -(n-1)K\right] \VHB_i \,,
}
$k_v^2$ takes the special value $k_v^2=(n-1)K$. Because $k_v^2\ge0$, this 
occurs only for $K=0$ or $K=1$. In the case $K=0$, this mode corresponds to 
the zero mode discussed above. In the case $K=1$, since we are assuming that 
$\K^n$ is complete, $\K^n$ is compact and closed, as known from Myers' 
theorem\cite{Myers.S1941}, and we can show the converse, i.e. that if 
$k_v^2=(n-1)K$, then $\VHB_{ij}$ vanishes, by integrating the identity 
$0= {{\VHB}^i} \hat{D}^j \VHB_{ij}
  =\hat{D}^j({{\VHB}^i}\VHB_{ij})+k_v {{\VHB}^{ij}}\VHB_{ij}$ over $\K^n$. 
Furthermore, using the same identities 
\Eqrsub{
&& 2D_{[i} V_{j]}D^{[i}V^{j]}=2D_i(V_jD^{[i}V^{j]})
   +V_j\insbra{-\triangle +(n-1)K} V^j \,,\\
&& 2D_{(i} V_{j)}D^{(i}V^{j)}=2D_i(V_jD^{(i}V^{j)})
   +V_j\insbra{-\triangle -(n-1)K} V^j \,.
\label{ids:Vector}
}
we can show that there is no eigenvalue in the range $0\le k_v^2\le(n-1)|K|$ where the second equality holds only for $K=-1$.

\subsubsection{Harmonic vectors on $S^n$}

We can give explicit expressions for vector harmonics on $S^n$ in terms of 
the homogeneous coordinates using Theorem~\ref{th:HarmonicScalar}.

\bigskip
\begin{theorem}[Harmonic vectors on $S^n$]\label{th:HarmonicVector}
Let us define the vector field $V_{\bm{b}}^i$ by
\Eq{
V_{\bm{b}}^i=b_{A_1\cdots A_\ell ;B}\Omega^{A_1}\cdots\Omega^{A_\ell }
                   \h D^i\Omega^B \,.
\label{HarmonicVector:HCR}
}
in terms of a constant tensor $\bm{b}=(a_{A_1\cdots A_\ell ;B})$($A_1,\cdots,A_\ell,B=1,\cdots,n+1$). Then, the following statements hold:
\Bitm
\item[1)] $V_{\bm{b}}^i$ is a divergence-free harmonic vector on $S^n$ with eigenvalue
\Eq{
k_v^2=\ell(\ell+n-1)-1,\quad \ell=1,2,\cdots \,,
}
if and only if the constant tensor $\bm{b}$ satisfies the conditions
\Eqrsubl{HarmonicVector:HCR:Conditions}{
&& b_{A_1\cdots A_\ell ;B}=b_{(A_1\cdots A_\ell );B} \,,
\label{HarmonicVector:HCR:C1}\\
&& b_{A_1\cdots A_{\ell -2}i}{}^i{}_{;B}=0 \,,
\label{HarmonicVector:HCR:C2} \\
&& b_{(A_1\cdots A_\ell ;A_{\ell +1})}=0 \,.
\label{HarmonicVector:HCR:C3}
}
\item[2)] All harmonic vectors $\Set{V_{\bm{b}}^i}$ form a complete basis 
in the $L^2$ space of divergence-free vector fields on $S^n$. 
Two harmonic vectors with different values of $\ell$ are orthogonal 
and those with the same $\ell$ have the inner product 
\Eq{
(V_{\bm{b}},V_{\bm{b}'})
  :=\int d^n\Omega (\bar V_{\bm{b}})_i V_{\bm{b}'}^i
  =C(n,\ell )\bm{\bar b}\cdot \bm{b}' \,, 
\label{HarmonicVector:HCR:InnerProduct}
}
where $C(n,\ell )$ is the number given in Eq.~\eqref{HarmonicFunction:HCR:InnerProduct}D
\item[3)] The multiplicity of the $\ell$-th eigenvalue is
\Eq{
N_V^l(S^n) 
 = \frac{(n+2\ell-1)(n+\ell-1)(n+\ell-3)!}{(\ell+1)(\ell-1)!(n-2)!} \,. 
}
\Eitm
\end{theorem}

\begin{proof}
For a harmonic vector $V^i$ with the eigenvalue $k_v^2$, let us define a set of functions on $S^n$ by $V^B=\h V^i \h D_i \Omega^B$. Then, from Eq.~\eqref{HarmonicCoord:Properties} we obtain $\h\triangle V^B =-(k_v^2+1)V^B$. This implies that $V^B$ is a harmonic function with the eigenvalue $k^2=k_v^2+1=\ell (\ell +n-1)$. Therefore, $V_i=V_B \h D_i\Omega^B$ can be written in terms of a constant tensor $b_{A_1\cdots A_\ell ;B}$ satisfying the conditions, Eqs.~\eqref{HarmonicVector:HCR:C1} and \eqref{HarmonicVector:HCR:C2} as in Eq.~\eqref{HarmonicVector:HCR}. The divergence of this expression can be written
\Eq{
\h D_iV^i =\left[-(n+\ell )b_{A_1\cdots A_\ell ;A_{\ell +1}}
     -\ell a_{A_1\cdots A_{\ell -1}}\delta_{A_\ell  A_{\ell +1}}\right]
     \Omega^{A_1}\cdots\Omega^{A_{\ell +1}} \,,
}
where
\Eq{
a_{A_1\cdots A_{\ell -1}}:=b_{A_1\cdots A_{\ell -1} i;}{}^i \,.
}
From this, it follows that the divergence free condition for $V^i$ can be expressed as
\Eqn{
(n+\ell )b_{(A_1\cdots A_\ell ;A_{\ell +1})}
=\ell a_{(A_1\cdots A_{\ell -1}}\delta_{A_\ell A_{\ell +1})} \,.
}
After a short caculation, we find that this condition is equivalent to Eq.~\eqref{HarmonicVector:HCR:C3}. This proves the first statement.

The completeness in the second statement immediately follows from the completeness of the harmonic vectors. The formula for the inner product can be derived by the same method as used for the harmonic scalars. 

Finally, we can show that the conditions, Eqs.~\eqref{HarmonicVector:HCR:Conditions}, on $\bm{b}$ are reduced to the conditions on $b_{A'_1\cdots A'_\ell;B'}$ and $b_{n+1 A'_1\cdots A'_{\ell-1};B'}$ with $A'_1,\cdots,A'_\ell,B'=1,\cdots,n$,
\Eqrsub{
&& b_{(A'_1\cdots A'_\ell;B')}=0 \,,\\
&& (\ell-2) b_{n+1}{}^{B'}{}_{B' (A'_1\cdots A'_{\ell-3};A'_{\ell-2})}
   + 3 b_{n+1\,A'_1\cdots A'_{\ell-3}}{}^{B'}{}_{;B'} \,,
}
and all the other components are uniquely determined from these components.  By counting the number of linealy independent solutions to these conditions, we obtain the multiplicity in the theorem.
\end{proof}

\subsection{Tensor harmonics}

\subsubsection{Definition and properties}

In the Einstein space $\K^n$, no condition is imposed directly on the Riemann tensor $R^i{}_{jkl}$ itself, though the Ricci tensor is assumed to be proportional to the metric. Hence, in general the Lichnerowicz operator appears in the tensor-type perturbation equation instead of a simple sum of the Laplace-Beltrami operator and scalars as $\K^n$-dependent part. Thus, we have to use the eigentensors to the Lichnerowicz operators to expand tensor-type perturbations on a generic Einstein space $\K^n$:
\Eqr{
&& \hat\triangle_L \THB_{ij}=\lambda_L\THB_{ij} \,, \\
&& \THB^i{}_i=0,\quad \h D^j \THB_{ij}=0 \,.
}
Note that the Lichnerowicz operator preserves the trace-free and transverse conditions. When $\K^n$ is a constant curvature space, $\THB_{ij}$ becomes a harmonic tensor on $\K^n$,
\Eq{
(\h\triangle + k_t^2)\THB_{ij}=0 \,, \quad k_t^2=\lambda_L -2nK \,.
}
Note also that an Einstein space with dimension equal to or smaller than $3$ 
is always a constant curvature space. Further, there exist no symmetric 
harmonic tensor with rank $2$ on $S^2$ and very special ones on $T^2$ and 
$H^2/\Gamma$. To be precise, the following theorem holds:

\medskip 
\begin{theorem}[Harmonic tensors on $2$-dimensional constant curvature space]
Let $\K$ be a two-dimensional closed surface with a constant curvature $K$. 
Then, a symmetric harmonic tensor $\THB_{ij}$ with rank 2 represents a moduli 
deformation of $\K$ and exists  only when $K\le0$. For $T^2$ ($K=0$), 
$\THB_{ij}$ is a constant trace-free tensor with $k^2=0$ in the chart in 
which $ds^2=dx^2+dy^2$, while  $k^2=-2K$ for $H^2/\Gamma$ ($K<0$). 
\end{theorem}

Very little is known about the spectrum of the Lichnerowicz operator on 
a generic Einstein space. However, we can easily show that 
\Eq{
k_t^2 \ge n|K| \,,
\label{value:kt2}
}
for tensor harmonics on a constant curvature space with the helps of the 
identities 
\Eqrsub{
&& 2D_{[i}T_{j]k} D^{[i}T^{j]k}= 2D^i(T_{jk}D^{[i}T^{j]k})
   +T_{jk}( -\triangle +nK) T^{jk} \,,\\
&& 2D_{(i}T_{j)k} D^{(i}T^{j)k}= 2D^i(T_{jk}D^{(i}T^{j)k})
   +T_{jk}( -\triangle -nK) T^{jk} \,.
}
%

\subsubsection{Harmonic tensors on $S^n$}

We can give explicit expressions for the symmetric harmonic tensors of 
rank $2$ on $S^n$ in terms of homogeneous coordiantes as in the case of 
harmonic vectors. 

\bigskip

\begin{theorem}[Harmonic tensors on $S^n$]
Let us define a tensor of rank 2 by 
\Eq{
T_{\bm{c}\,ij}=c_{A_1\cdots A_\ell ;B_1 B_2}\Omega^{A_1}\cdots\Omega^{A_\ell }
       \h D_i\Omega^{B_1}\h D_j\Omega^{B_2}
\label{Harmonic2Tensor:HCR}
}
where $\bm{c}=(c_{A_1\cdots A_\ell ;B_1 B_2})$ 
($A_1,\cdots,A_\ell,B_1,B_2=1,\cdots,n+1$) is a constant tensor on $E^{n+1}$. 
Then, the following statements hold: 
\Bitm
\item[i)] $T_{\bm{c}\,ij}$ is a harmonic tensor with the eigenvalue
\Eq{
k^2=\ell (\ell +n-1)-2\ (\ell =2,\cdots)
}
if and only if $\bm{c}$ satisfies the following conditions:
\Eqrsubl{Harmonic2Tensor:HCR:C}{
&& c_{A_1\cdots A_\ell ;B_1 B_2}=c_{(A_1\cdots A_\ell );B_1 B_2} \,,
\label{Harmonic2Tensor:HCR:C1}\\
&& c_{A_1\cdots A_{\ell -2}A}{}^A{}_{;B_1 B_2}=0 \,,
\label{Harmonic2Tensor:HCR:C2}\\
&& c_{(A_1\cdots A_\ell ;A_{\ell +1})}{}^B
   = c_{(A_1\cdots A_\ell ;}{}^B{}_{A_{\ell +1})}=0 \,.
\label{Harmonic2Tensor:HCR:C3}
}
\item[2)] The set of all harmonic tensors of the form $T_{\bm{c}\,ij}$ forms 
a complete basis for the $L^2$ space of trace-free and divergence-free tensors 
of rank $2$ on $S^n$. Two such harmonic tensors with different values of 
$\ell$ are orthogonal and those with the same $\ell$ have the innter product 
\Eq{
(T_{\bm{c}},T_{\bm{c}'}):=\int d^n\Omega \bar T_{\bm{c}\,ij}
   T_{\bm{c}'}{}^{ij}
   =C(n,\ell )\bm{\bar c}\cdot\bm{c}' \,,
\label{Harmonic2Tensor:HCR:InnerProduct}
}
where $C(n,\ell )$ is the same constant as that in Eq.~\eqref{HarmonicFunction:HCR:InnerProduct}D
\item[3)] For  $\ell\ge2, n\ge2$, the multiplicity of the $\ell$-th eigenvalue for symmetric harmonic tensors of rank 2 is 
\Eq{
N^\ell_T(S^n) =\frac{(n+1)(n-2)(n+\ell)(n+2\ell-1)(n+\ell-3)!}{2\ell(\ell-2)!(n-1)!} \,. 
}
\Eitm
\end{theorem}

\begin{proof}
The first two statements can be easily proved by methods similar to 
those for the harmonic vector. To prove the last statement, lengthy 
calculations or sophisticated considerations based on group representation 
theory are required. See Ref.~\citen{Kanti.P&&2010} for details. \end{proof}


\section{Tensor-type perturbations} 
\label{sec:TensorPerturbation}

\subsection{Generic background} 
We first consider tensor perturbations in an $(m+n)$-dimensional generic 
background metric, Eq.~\eqref{BG:metric}, and the energy-momentum tensor,  
Eq.~\eqref{BG:energymomentum}. 
We have already seen in the previous section that 
$h_T^{(2)}{}_{ij}$, $\delta T_T^{(2)}{}_{ij}$ are
by themselves gauge-invariant. We expand these two in terms of the 
eigentensors $\THB_{ij}$ of the Lichnerowicz operator $\hat\triangle_L$ 
introduced in the previous section as follows: 
\Eq{
h_T^{(2)}{}_{ij} = 2r^2 H_T \THB_{ij}\,,
\quad 
\delta T_T^{(2)}{}_{ij} = r^2\left( \tau_T+2PH_T \right)\THB_{ij} \,. 
\label{TensorPerturbation:metric}  
}
The Einstein equations in terms of the gauge-invariant coefficients 
$H_T$ and $\tau_T$ are obtained from Eq.~(23) in 
Ref.~\citen{Kodama.H&Ishibashi&Seto2000} with the replacement 
$k^2 \tend \lambda_L-2nK$, where $\lambda_L$ is 
the eigenvalue of the Lichnerowicz operator. 
The result is expressed by the single equation 
\Eq{
\Box H_T +\frac{n}{r}Dr\cdot 
DH_T-\frac{\lambda_L-2(n-1)K}{r^2}H_T=-\kappa^2\tau_T \,. 
\label{BasicEq:tensor}
}

We emphasise that this equation holds whenever the background metric
is given in the form of Eq.~\eqref{BG:metric}, irrespective
to the dimension of $\N^m$, and therefore applies to a more general
background than that of a static black hole. For example, it has been applied 
to the stability analysis\cite{Kodama.H&Konoplya&Zhidenko09} 
of Myers-Perry black holes with a single rotation,  
whose metric takes the warped product form of Eq.~\eqref{BG:metric}, 
with $\K^n$ being the $(d-4)$-dimensional unit sphere.

We also note that Eq.~\eqref{BasicEq:tensor} is precisely the same form as 
the equation of motion for a massless test scalar field on the same 
background spacetime if $H_T$ is viewed as the scalar field with 
the angular momentum number being $\lambda_L-2(n-1)K$. 

\subsection{Static black hole background} 
We turn to the black hole background, Eq.~\eqref{metric:GSBH}, for which 
$m=2$. If one introduces the master variable $\Phi$ by 
\Eq{
\Phi=r^{n/2}H_T,
}
Eq.~\eqref{BasicEq:tensor} can be put into the canonical form  
\Eq{
 \Box\Phi- \frac{V_T}{f}\Phi=-\kappa^2 r^{n/2} \tau_T, 
\label{MasterEq:Tensor}
}
where
\Eq{
V_T= \frac{f}{r^2}\left[ 
                        \lambda_L-2(n-1)K
                        +\frac{nrf'}{2}+\frac{n(n-2)f}{4}
                  \right]. 
\label{VT:GSBH}
}
In particular, for $f(r)$ given by Eq.~\eqref{f:RNBH}, $V_T$ is expressed as
\Eq{
V_T=\frac{f}{r^2}\left[\lambda_L +\frac{n^2-10n+8}{4}K 
    -\frac{n(n+2)}{4}\lambda r^2+\frac{n^2M}{2r^{n-1}}
    -\frac{n(3n-2)Q^2}{4r^{2n-2}}
 \right].
\label{VT:RNBH}
}

Note that since an electromagnetic field $\F_{ab}$ is described 
by a vector field, it does not have any tensor-type component. 
The electromagnetic field enter the equations for a tensor perturbation 
only through their effect on the background geometry.

\section{Vector-type perturbations}  
\label{sec:VectorPerturbation}

\subsection{General background case}  
We have already introduced a basis of vector-type gauge-invariant variables 
$(F^{(1)}_{ai}, \tau^{(1)}_{ai}, \tau^{(1)}_{ij} )$ in generic background 
of $(m+n)$-dimensions. 
We expand these gauge invariants in terms of vector harmonics $\VHB_i$ 
and write the perturbed Einstein equations for the expansion coefficients.  

\smallskip
For generic modes $m_V := k_v^2-(n-1)K\not=0$, 
\Eqr{
  F^{(1)}_{ai} = rF_a \VHB_i \,, \quad 
  \tau^{(1)}_{ai} = r\tau_a \VHB_i \,, \quad 
  \tau^{(1)}_{ij} = r^2 \tau_T \VHB_{ij} \,,  
\label{def:vector:Fai:tauai:tauij}
}
and the Einstein equations reduce to 
\Eqr{
&& 
\frac{1}{r^{n+1}}D^b
     \left\{
         r^{n+2}\left[
                      D_b\left(\frac{F_a}{r}\right)
                     -D_a\left(\frac{F_b}{r}\right)
                \right]
     \right\} - \frac{m_V}{r^2}F_a
    = -2\kappa^2 \tau_a \,, 
\label{eq:Vectorperturbation:evol:Fa}
\\
&&  
 \frac{k}{r^n}D_a(r^{n-1}F^a) = -\kappa^2\tau_T \,.
\label{eq:Vectorperturbation:constr:Fa}
}

\smallskip 
For the exceptional mode $m_V=0$, only the following combination 
\Eqr{
F^{(1)}_{ab}=rD_a\left(\frac{F_b}{r}\right)-rD_b\left(\frac{F_a}{r}\right)\,,
}
is gauge-invariant. For this mode we have only a single equation
\Eq{
 \frac{1}{r^{n+1}}D^b(r^{n+1}F^{(1)}_{ab})=-2\kappa^2\tau_a \,. 
}

\subsection{Static black hole background case}  

In the black hole case with $m=2$, using the $2$-dimensional Levi-Civita 
tensor $\epsilon_{ab}$, we can rewrite the perturbed Einstein equations as 
\Eqrsubl{VectorPerturbation:BasicEq:metric}{
&&D_a\left( r^{n+1}F^{(1)} \right)
  -m_V r^{n-1}\epsilon_{ab}F^b
  =-2\kappa^2 r^{n+1}\epsilon_{ab}\tau^b,
\label{BasicEq:metric:vector1}\\
&& k_v D_a(r^{n-1}F^a)=-\kappa^2r^n\tau_T,
\label{BasicEq:metric:vector2}
}
where 
\Eq{
F^{(1)}=\epsilon^{ab}r D_a\pfrac{F_b}{r} \,,
\label{F^(1):def}
}
and $F_a$ is defined by the first of Eq.~\eqref{def:vector:Fai:tauai:tauij} 
with $F^{(1)}_{ai} \rightarrow h^{(1)}_{ai}$. 
This should be supplemented by the perturbation of the 
energy-momentum conservation law
\Eq{
D_a(r^{n+1}\tau^a)+\frac{m_V}{2k_v}r^n\tau_T=0 \,.
\label{EMconservation:vector}
}

Note that, for $m_V=0$, the perturbation variables $h_T^{(1)}{}_i$ and 
$\delta T_T^{(1)}{}_i$---hence $H_T$ and $\tau_T$---do not exist. 
The matter variable $\tau_a$ is still gauge-invariant, but concerning
the metric variables, only the combination $F^{(1)}$ defined in 
Eq.~\eqref{F^(1):def} is gauge invariant. In this case, the Einstein equations 
are reduced to the single equation \eqref{BasicEq:metric:vector1}, and 
the energy-momentum conservation law is given by 
Eq.~\eqref{EMconservation:vector} without the $\tau_T$ term.

These gauge-invariant perturbation equations can be reduced to 
a single wave equation with a source in the $2$-dimensional spacetime 
$\N^2$. First, for the generic modes $m_V\not=0$,  
from Eqs.~\eqref{EMconservation:vector} and \eqref{BasicEq:metric:vector2} 
we find that $F^a$ can be written in terms of a variable $\tilde \Omega$ as 
\Eq{
 \epsilon^{ab}D_b\tilde\Omega = 
r^{n-1}F^a - \frac{2\kappa^2}{m_V}r^{n+1}\tau^a \,.  
\label{FaByOmega}
}
Inserting this expression into Eq.~\eqref{BasicEq:metric:vector1}, we 
obtain the master equation
\Eq{
r^nD_a\left( \frac{1}{r^n}D^a\tilde\Omega \right)
  -\frac{m_V}{r^2}\tilde\Omega 
  =-\frac{2\kappa^2}{m_V} r^n\epsilon^{ab}D_a(r\tau_b) \, . 
\label{MasterEq:vector:metric}
}

For the special modes $m_V=0$, 
it follows from Eq.~\eqref{EMconservation:vector} with $\tau_T=0$ that 
$\tau_a$ can be expressed in terms of a function $\tau^{(1)}$ as
\Eq{
r^{n+1} \tau_a=\epsilon_{ab}D^b\tau^{(1)}.
\label{tau1:def}
}
Inserting this expression into Eq.~\eqref{BasicEq:metric:vector1} with 
$\epsilon^{cd}D_c(F_d/r)$ replaced by $F^{(1)}/r$, we obtain
\Eq{
D_a(r^{n+1}F^{(1)})=-2\kappa^2D_a\tau^{(1)}.
} 
Taking into account of the freedom of adding a constant in the definition 
of $\tau^{(1)}$, we have the general solution 
\Eq{
F^{(1)}=-\frac{2\kappa^2 \tau^{(1)}}{r^{n+1}}.
\label{BasicEq:vector:metric:exceptional}
}
Hence, there exists no dynamical freedom in these special modes. In 
particular, in the source-free case in which $\tau^{(1)}$ is a 
constant and $K=1$, this solution corresponds to adding a small 
rotation to the background solution.

\subsection{Static black hole in Einstein-Maxwell system}

\subsubsection{Perturbation of electromagnetic fields}
\label{subsec:EMFperturbation}

For vector-type perturbations, we note that ${\delta\F}_{MN}$ 
is gauge-invariant as well as invariant under a coordinate 
gauge transformation.  
As shown in Eq.~\eqref{Pert:Vector:Maxwell:A}, 
a vector perturbation of the Maxwell field can be expressed, 
in terms of the single gauge-invariant variable $\A$ defined 
by $\delta \A^{(1)}_i =\A \VHB_i$ on $\N^2$, as 
\Eq{ 
\delta \F_{ab}=0 \,,\ 
\delta \F_{ai}=D_a\A\VHB_i \,,\ 
\delta \F_{ij}=\A \left(\hat D_i\VHB_j-\hat D_j\VHB_i\right) \,.  
} 
We also expand the current $J_i$ as 
\Eq{
J_i=J\VHB_i \,.   
}
Then we obtain from the second of Eq.~\eqref{MaxwellEq:Perturbation:Ji} 
the gauge-invariant form for the Maxwell equation,  
\Eq{
 \frac{1}{r^{n-2}}D_a(r^{n-2}D^a \A)
    -\frac{k_v^2+(n-1)K}{r^2}\A  
   =-J +r E_0 F^{(1)} \,. 
\label{RN:vector:Master:EM0}
}

In order to complete the formulation of the basic perturbation equations, 
we must separate the contribution of the electromagnetic field to the source 
term in the Einstein equation \eqref{MasterEq:vector:metric}. 
%
The contributions of the electromagnetic field to $\tau_a$ and 
$\tau_T$ are given in terms of $\A$ by
\Eq{
\tau_{a}^{{\rm (em)}}=-\frac{E_0}{r}\epsilon_{ab}D^b\A \,, 
\quad 
\tau_{T}^{{\rm (em)}}=0 \,.
\label{tau:vector:EM}
}
Hence, the Einstein equations for the Einstein-Maxwell system can be 
obtained by replacing $\tau_a$ in Eq.~\eqref{MasterEq:vector:metric} by
\Eq{
\tau_a=\tau_{a}^{{\rm (em)}}+ \bar\tau_a \,, 
}
where the second term represents the contribution from matter other 
than the electromagnetic field. 

\subsubsection{Master equations}
Now, generic modes: $m_V \neq 0$, 
we introduce new master variables by 
\Eq{
 \Phi_\pm:=a_\pm r^{-n/2}\left(\tilde\Omega -\frac{2\kappa^2 q}{m_V}\A \right) 
          + b_\pm r^{n/2-1}\A \,,
}
with 
\Eqrsub{
&& (a_+,b_+)= 
  \left( \frac{Q m_V}{(n^2-1)M+\Delta},\frac{Q}{q} \right) \,, 
\\
&& (a_-,b_-)= \left(1,\frac{-2n(n-1)Q^2}{q[(n^2-1)M+\Delta]} \right) \,, 
}
where $\Delta$ is a positive constant satisfying 
\Eq{
\Delta^2= (n^2-1)^2M^2+2n(n-1)m_VQ^2 \,. 
}
Then, from Eqs.~\eqref{MasterEq:vector:metric} 
and \eqref{RN:vector:Master:EM0}, 
we obtain the two decoupled wave equations as our master equations: 
\Eq{ 
  \Box \Phi_\pm - \frac{V_{V\pm}}{f}\Phi_\pm =S_{V\pm} \,,  
}
%
where  
\Eq{
V_{V\pm}=\frac{f}{r^2}\left[k_v^2 +\frac{(n^2-2n+4)K}{4}
   -\frac{n(n-2)}{4}\lambda r^2+\frac{n(5n-2)Q^2}{4r^{2n-2}} 
   +\frac{\mu_\pm}{r^{n-1}}\right] \,,
\label{Vpm:Vector}
}
\Eq{ 
\mu_\pm=-\frac{n^2+2}{2}M \pm \Delta \,, 
}
and 
\Eq{
S_{V\pm}=-a_\pm \frac{2\kappa^2 r^{n/2}f}{m_V}
      \epsilon^{ab}D_a(r\bar\tau_b)
    -b_\pm r^{n/2-1}f J \,. 
}
For $n=2, K=1$ and $\lambda=0$, the variables $\Phi_+$ and $\Phi_-$ 
are proportional to the variables for the axial modes, $Z^{(-)}_1$ 
and $Z^{(-)}_2$ given in Ref.~\citen{Chandrasekhar.S1983B}, and 
$V_{V+}$ and $V_{V-}$ coincide with the corresponding potentials, 
$V^{(-)}_1$ and $V^{(-)}_2$, respectively.

Here, note that in the limit $Q\tend0$, $\Phi_+$ becomes 
proportional to $\A$ and $\Phi_-$ to $\Omega$. Hence, $\Phi_+$ and 
$\Phi_-$ represent the electromagnetic mode and the gravitational 
mode, respectively. In particular, in the limit $Q\tend0$, the 
equation for $\Phi_-$ coincides with the master equation for a 
vector perturbation on a neutral black hole background derived in 
Paper~I. 

As for the exceptional modes, $m_V=0$, from 
the definition, Eq.~\eqref{tau1:def}, of $\tau^{(1)}$ and 
Eq.~\eqref{tau:vector:EM}, we can express $\tau^{(1)}$ as 
\Eq{
\tau^{(1)}=-q\A +\bar\tau^{(1)} \,. 
}
Hence, Eq.~\eqref{BasicEq:vector:metric:exceptional} can be rewritten as
\Eq{
 F^{(1)}=\frac{2\kappa^2(q\A-\bar\tau^{(1)})}{r^{n+1}} \,.
}
Inserting this into Eq.~\eqref{RN:vector:Master:EM0}, we obtain
\Eq{
\frac{1}{r^{n-2}}D_a(r^{n-2}D^a\A)-\frac{1}{r^2}
 \left(2(n-1)K + \frac{2n(n-1)Q^2}{r^{2n-2}}\right)\A 
= -J -\frac{2\kappa^2 q}{r^{2n}}\bar\tau^{(1)} \,. 
\label{RN:vector:Master:EM:exceptional}
}
Therefore, only the electromagnetic perturbation is dynamical.

\section{Scalar-type perturbations} 
\label{sec:ScalarPerturbation}

\subsection{General background case}
In terms of scalar harmonics $\SHB$, we expand the scalar-type 
perturbation variables as 
\Eqrsub{
 && h_{ab}=f_{ab}\SHB \,, \quad
 h_a = - \frac{r}{k}f_a \SHB \,, \quad
 h_L = 2r^2H_L \SHB\,, \quad
 h_T = 2r^2\frac{H_T}{k^2} \SHB\,,
\\
 &&
 \delta T_{ab}= \tau_{ab}\SHB \,, \quad
 \delta T_a= - \frac{r}{k}(Pf_a +\tau_a)\SHB \,, \quad
\\
 &&
 \delta T_L =r^2(2H_LP +\delta P)\SHB \,, \quad
 \delta T_T =\frac{r^2}{k^2}(2H_TP +\tau_T)\SHB \,, 
}
and $X_a=X_a \SHB$. Here $P$ and all expansion coefficients, 
e.g. $\delta P$, are tensor fields on the $m$-dimensional spacetime $\N^m$. 

The basis of the scalar-type gauge-invariant variables introduced
in the previous section are expanded as:  
\Eqrsub{
&&F^{(0)}{}_{ab} = F_{ab}\SHB \,, \quad  
  F^{(0)} = 2r^2F \SHB \,, 
\\
&&
 \Sigma^{(0)}{}_{ab}=\Sigma_{ab} \SHB \,, \quad
 \Sigma^{(0)}{}_{ai}=r\Sigma_a \SHB_i \,, \quad
\\
&&
 \Sigma^{(0)} = r^2 \Sigma \SHB \,, \quad 
 \Pi^{(0)}{}_{ij} = r^2 \tau_T \SHB_{ij} \,.
 }
The expansion coefficients here 
$F_{ab},F,\Sigma_{ab},\Sigma_a,\Sigma,\tau_T$ as well as $X_a$ 
are precisely the same as those given in 
Ref.~\citen{Kodama.H&Ishibashi&Seto2000}.

For the exceptional modes with $k^2=n$ for $K=1$,
$h_T$ and $\delta T_T$ (equivalently $H_T$ and $\tau_T$) are 
not defined, because a second-rank symmetric tensor cannot be
constructed from $\SHB$ for these modes. In this case, 
we define $F, F_{ab}, \Sigma_{ab},\Sigma_a$ and 
$\Sigma_L$ by setting $X_L=0$ in the above definitions. These 
quantities defined in this way are, however, gauge dependent. These 
exceptional modes are treated in Appendix~C of Paper~III.

\subsection{Static black hole background}

\subsubsection{Maxwell equations}
From now on we consider the static black hole background with $m=2$. 
We often use $\epsilon_{ab}$.  
We expand $\E^{(0)}$ and $\E^{(0)}_a$ defined in 
Eqs.~\eqref{def:gi:scalar:pert:Maxwell}, \eqref{def:gi:scalar:pert:Maxwell:a} as  
\Eqrsubl{GaugeInvVar:scalar:EMF}{
 {\E} \SHB = {\E}^{(0)} \,, \quad 
 {\E}_a \SHB = -\frac{k}{r}{\E}^{(0)}_a \,. 
}
Then, the Maxwell equations \eqref{MaxwellEq:Perturbation} are written 
\Eqrsub{
&& \frac{1}{r^n}D_a(r^n\E)+\frac{k}{r}\E_a
  -\frac{E_0}{2}D_a(F^c_c-2nF)=\epsilon_{ab}J^b,
\label{BasicEq:EM:Scalar:Maxwell1}\\
&& \epsilon^{ab}D_a(r^{n-1}\E_b)=-r^{n-1}J,
\label{BasicEq:EM:Scalar:Maxwell2}
}
with $J$ defined by $J_i=rJ \SHB_i$. 
Note that from the first of Eqs.~\eqref{MaxwellEq:Perturbation:dF} we have 
the relation 
\Eq{
\E=-\frac{1}{k}D_c(r\E^c).
\label{BasicEq:EM:Scalar:EbyEa}
}
Note also that Eqs.~\eqref{BasicEq:EM:Scalar:Maxwell1} and 
\eqref{BasicEq:EM:Scalar:Maxwell2} give the current conservation law
\Eq{
D_c(r^n J^c)=-kr^{n-1}J.
\label{BasicEq:EM:Scalar:J}
}

We find that the gauge invariant variables, $\E_a$ and $\E$, 
can be expressed in terms of the single master variable $\A$ as 
\Eqr{
\E_a =\frac{k}{r^{n-1}}\left(D_a \A + \tilde J_a\right) \,, 
\label{BasicEq:EM:Scalar:EabyA}  
\quad 
r^n\E = -k^2\A + \frac{q}{2}(F^c_c-2nF) \,, 
\label{BasicEq:EM:Scalar:EbyA}
}
where $\tilde J_a $ has been
defined by $J^a= {k^2}{r^{-n}}\epsilon^{ab}\tilde J_b$. 
By inserting these expressions into Eq.~\eqref{BasicEq:EM:Scalar:EbyEa}, 
we obtain the wave equation for $\A$: 
\Eq{
r^{n-2}D_a\left( \frac{D^a\A}{r^{n-2}} \right) -\frac{k^2}{r^2}\A
  =-r^{n-2}D_a\left(\frac{\tilde J^a}{r^{n-2}}\right)
    -\frac{q}{2r^2}(F^c_c-2nF) \,.  
\label{BasicEq:EM:Scalar:A}
}

The contribution of the electromagnetic field to the perturbation of the 
energy-momentum tensor, $\Sigma_{ab}$, $\Sigma_a$ and $\Sigma_L$,
are 
\Eqrsub{
&& \Sigma_{ab}^{{\rm (em)}}=\left( \frac{qk^2}{r^{2n}}\A
  +\frac{nq^2}{r^{2n}}F \right)g_{ab}
 -\frac{q^2}{2r^{2n}}F_{ab} \,,\\
&& \Sigma_{a}^{{\rm (em)}}
    =-\frac{qk}{r^{2n-1}}\left(D_a\A+\tilde J_a\right) \,,\\
&& \Sigma^{{\rm (em)}}= -\frac{qk^2}{r^{2n}}\A-\frac{nq^2}{r^{2n}}F \,.  
\label{Sigma:EMF}
}

\subsubsection{Master equations}

For generic modes of scalar perturbations, the Einstein equations 
consist of four sets of equations of the forms 
\Eq{
E_{ab}=\kappa^2 \Sigma_{ab} \,, \ 
E_a=\kappa^2 \Sigma_a \,,\ 
E_L=\kappa^2 \Sigma \,,\ 
E_T=\kappa^2 \tau_T \,.
\label{eqn:Einstein:KIS00}
}
For the definitions of $E_{ab}, E_a, E_L$ and $E_T$, see 
eqs.~(63)--(66) in Ref.~\citen{Kodama.H&Ishibashi&Seto2000}.  
Introducing the perturbation of the energy-momentum tensor as 
\Eq{
S_{ab}=r^{n-2}\kappa^2(\Sigma_{ab}-\Sigma_{ab}^{{\rm (em)}}),\ 
S_a=\frac{r^{n-1}\kappa^2}{k}(\Sigma_{a}-\Sigma_{a}^{{\rm (em)}}),\ 
S_L=r^{n-2}\kappa^2(\Sigma_{L}-\Sigma_{L}^{{\rm (em)}}) \,, 
\label{S:def}
} and 
\Eq{
S_T=\frac{2r^n}{k^2}\kappa^2 \tau_T.
\label{S_T:def}
}  
the conservation law for the energy-momentum tensor is given by 
the two equations
\Eqrsubl{EMconservation:scalar:RNBH}{
&& \frac{1}{r^2}D_a(r^2 S^a) -S_L+\frac{(n-1)(k^2-nK)}{2nr^2}S_T=0,\\
&& \frac{1}{r^2}D_b(r^2S^b_a)+\frac{k^2}{r^2}S_a -n\frac{D_a r}{r}S_L
   =k^2\frac{\kappa^2 q}{r^{n+2}}\tilde J_a.
}

Now we introduce $X,Y$ and $Z$ defined by
\Eq{
    X= r^{n-2}(F^t_t-2 F),\ 
    Y= r^{n-2}(F^r_r-2 F),\ 
    Z= r^{n-2} F^r_t,
\label{XYZ:def}
}
as in Paper~I. 
After the Fourier transformation with respect to the Killing time coordinate,  
$t$, of the black hole background, Eq.~\eqref{metric:GSBH}, 
the perturbation equations above can be reduced to a system consisting 
of three first order linear differential equations for $X,Y,Z$ and a single 
linear algebraic constraint on them, with inhomogeneous terms given by 
$\A,J_a, S_{ab}, S_a,S_T$. As performed in Paper~I, this constrained system 
can be further simplified to a single second-order ODE for a scalar field 
$\Phi$ given by 
\Eq{
 \Phi= -\frac{X+Y+S_T- {nZ}/{i\omega r}}{r^{n/2-2}H} \,,
\label{RN:scalar:MasterVar:metric}
}
with a source term, where 
\Eqr{
&& H=m +\frac{n(n+1)M}{r^{n-1}} - \frac{n^2Q^2}{r^{2n-2}} \,,\\ 
&& m=k^2-nK \,. 
}
Thus obtained ODE for $\Phi$ and the Maxwell equation for $\A$, 
given by Eq.~\eqref{BasicEq:EM:Scalar:A}, form a coupled second-order 
ODEs with source terms. Then, by introducing new master variables, 
$\Phi_\pm$, given as a linear combination of $\Phi$ and $\A$ below, 
we obtain the following two decoupled master equations for scalar-type 
perturbations:  
\Eq{
 f\frac{d}{dr}\left(f\frac{d}{dr}\Phi_\pm \right)
 + \left(\omega^2-V_{S\pm}\right)\Phi_\pm   =S_{S\pm} \,,   
}
where 
\Eq{
 \Phi_\pm :=a_\pm \Phi + b_\pm \A \,,
}
with 
\Eqrsub{
&& (a_+,b_+)=\left(\frac{m}{n}Q+\frac{(n+1)(M+\mu)}{2r^{n-1}}Q \,,
    \frac{(n+1)(M+\mu)Q}{qr^{n/2-1}}\right),\\
&& (a_-,b_-)=\left((n+1)(M+\mu)-\frac{2nQ^2}{r^{n-1}} \,,
    -\frac{4nQ^2}{qr^{n/2-1}} \right) \,, 
}
and where $\mu$ is a positive constant satisfying 
\Eq{
\mu^2=M^2+\frac{4mQ^2}{(n+1)^2} \,. 
}
If we define the parameter $\delta$ by 
\Eq{
\mu=(1+2m\delta)M \,,
}
the effective potentials $V_{S\pm}$ are expressed as
\Eq{
V_{S\pm} =\frac{fU_\pm}{64r^2H_\pm^2} \,,
}
with 
\Eqrsub{
&& H_+=1-\frac{n(n+1)}{2}\delta x \,, \quad 
   H_-=m+\frac{n(n+1)}{2}(1+m\delta)x \,, 
\\
&& 
 x=\frac{2M}{r^{n-1}} \,,\quad 
 y=\lambda r^2 \,, \quad 
 z=\frac{Q^2}{r^{2n-2}} \,,      
}
and 
\Eqrsub{
& U_+ =
& \left[-4 n^3 (n+2) (n+1)^2 \delta^2 x^2-48 n^2 (n+1) (n-2) \delta x
\right.\notag\\ 
&&\left.   -16 (n-2) (n-4)\right] y
  -\delta^3 n^3 (3 n-2) (n+1)^4 (1+m \delta) x^4
\notag\\
&&   +4 \delta^2 n^2 (n+1)^2 
   \left\{(n+1)(3n-2) m \delta+4 n^2+n-2\right\} x^3
\notag\\   
&&   +4 \delta (n+1)\left\{
   (n-2) (n-4) (n+1) (m+n^2 K) \delta-7 n^3+7 n^2-14 n+8
   \right\}x^2
\notag\\   
&&  + \left\{16 (n+1) \left(-4 m+3 n^2(n-2) K\right) \delta
     -16 (3 n-2) (n-2) \right\}x
\notag\\   
&&    +64 m+16 n(n+2) K \,,\\
& U_- =
  & \left[-4 n^3 (n+2) (n+1)^2 (1+m \delta)^2 x^2
      +48 n^2 (n+1) (n-2) m (1+m \delta) x  \right.
\notag\\
&& \left.  -16 (n-2) (n-4) m^2\right] y
     -n^3 (3 n-2) (n+1)^4 \delta (1+m \delta)^3 x^4
\notag\\
&& -4 n^2 (n+1)^2 (1+m \delta)^2 
     \left\{(n+1)(3 n-2) m \delta-n^2\right\} x^3
\notag\\  
&&  +4 (n+1) (1+m \delta)\left\{ m (n-2) (n-4) (n+1) (m+n^2 K) \delta
  \right. \notag\\
&& \left. \quad  +4 n (2 n^2-3 n+4) m+n^2 (n-2) (n-4) (n+1)K \right\}x^2
\notag\\
&&  -16m \left\{ (n+1) m \left(-4 m+3 n^2(n-2) K\right) \delta
\right.\notag\\
&&\left.  +3 n (n-4) m+3 n^2 (n+1) (n-2)K \right\}x
\notag\\
&&      +64 m^3+16 n(n+2)m^2 K \,.  
}

The source terms $S_{S\pm}$ are given by 
%
\Eq{
S_{S\pm}=a_\pm  \bar S_\Phi + b_\pm  S_\A \,, 
}
where $\bar S_\Phi=S_\Phi|_{\A=0}$ with 
\Eqr{
&& S_\Phi=\frac{f}{r^{n/2}H}\left[ \kappa^2 E_0\left( 
      \frac{P_{S1}}{H}\left(\A-\frac{\tilde J_t}{i\omega}\right)
      +2nrf \tilde J_r 
   +2k^2\frac{\tilde J_t}{i\omega}
   +2nf\frac{r\partial_r \tilde J_t}{i\omega} \right) \right.
   \notag\\
&& \qquad\qquad\qquad
   -HS_T -\frac{P_{S2}}{H}\frac{S_t}{i\omega }
   -2nf\frac{r\partial_r S_t}{i\omega}
   -2nrfS_r \notag\\
&& \qquad\qquad\qquad \left.
   +\frac{P_{S3}}{H}\frac{rS^r_t}{i\omega}
   +2r^2\frac{\partial_r S^r_t}{i\omega}
   +2r^2S^r_r
   \right] \,,
\label{RN:scalar:MasterEq:Source:Phi}
}
and 
\Eq{
S_A= -\left( \frac{2n^2(n-1)^2zf^2}{r^2H}+\omega^2 \right)
        \frac{\tilde J_t}{i\omega}
        -r^{n-2}f\partial_r\pfrac{f\tilde J_r}{r^{n-2}}
        +\frac{2(n-1)E_0}{i\omega H}f\left( nf S_t
        -r S^r_t \right) \,.
\label{RN:scalar:MasterEq:Source:A}
}
Here $P_{S1}, P_{S2}$ and $P_{S3}$ are polynomials of $x,y$ and $z$, 
whose explicit expressions are 
\Eqr{
& P_{S1}=& \left[ -4n^4z+2n^2(n+1)x-4n(n-2)m \right]y \notag\\
&& +\left\{ 2n^2(n-1)x+4n(n-2)m+4n^3(n-2)K \right\}z
 \notag\\
&& -n^2(n^2-1)x^2+\left\{ -4n(n-2)m+2n^2(n+1)K \right\}x
\notag\\
&& +4m^2+4n^2mK \,, 
\\
& P_{S2}=& \left[ 6n^4z-n^2(n+1)(n+2)x+2n(n-4)m \right]y
\notag 
\\
&& -2n^4 z^2+\left\{ n^2(3n^2-n+2)x-4n(n-2)m 
   -6n^3(n-1)K \right\}z \notag 
\\
&&   -n^2(n+1)x^2
   +\left\{ n(3n-7)m+n^2(n^2-1)K \right\}x
  \notag\\
&& -2m^2-2n(n-1)mK \,, 
\\
& P_{S3}=& -2n^2(3n-2)z+n^2(n+1)x-2(n-2)m \,. 
}
Note that the master variable $\Phi_-$ coincides with that for the neutral 
and source-free case in Paper~I for $Q=0$ and $S_T=0$. 
Also note that the following relations hold: 
\Eqr{
&& Q^2=(n+1)^2M^2 \delta( 1+m\delta) \,, 
\\ 
&& H=H_+ H_- \,.
}
%


{}From these relations, we find that $Q=0$ corresponds to $\delta=0$, and 
in this limit, $\Phi_-$ coincides with $\Phi$, and its equation coincides 
with the mater equation for the master variable $\Phi$ derived in Paper~I. 
Hence, $\Phi_-$ and $\Phi_+$ represent the gravitational mode and the 
electromagnetic mode, respectively. For $n=2, K=1$ and $\lambda=0$, these 
variables $\Phi_+$ and $\Phi_-$ are proportional to the variables for the 
polar modes, $Z^{(+)}_1$ and $Z^{(+)}_2$, appearing in 
Ref.~\citen{Chandrasekhar.S1983B}, and $V_{S+}$ and $V_{S-}$ coincide 
with the corresponding potentials, $V^{(+)}_1$ and $V^{(+)}_2$, 
respectively.

For the exceptional modes, 
the last of Eq.~\eqref{eqn:Einstein:KIS00} is not obtained 
from the Einstein equations. However, this equation with $\tau_T=0$ 
can be imposed as a gauge condition, as shown in Paper~I. Under this 
gauge condition, all equations derived in this subsection hold without 
change. However, the variables still contain some residual gauge 
freedom. See Appendix~D of Paper~III, for how to eliminate 
the residual gauge freedom to extract physical degrees of freedom.


\section{Lovelock Black Holes}
\label{sec:Lovelock}

When Einstein derived the field equations for gravity, he adopted the three 
requirements as the guiding principle in addition to the principle of general 
relativity. The first is the metric ansatz that gravitational field is completely 
determined by the spacetime metric. The second is that the energy-momentum 
tensor sources gravity, hence it is balanced by a second-rank symmetric tensor 
constructed from the metric. The last is the requirement that the field 
equations contain the second derivatives of the metric at most and are 
quasi-linear, i.e., the coefficients of the second derivatives contain only 
the metric and its first derivatives. These requirements determine the 
gravitational field equation uniquely. 

In $4$-dimensions, we can obtain a similar result even if we loosen the third 
condition and require only that the field equations contain second derivatives 
of the metric at most. To be precise, gravity theories satisfying these weaker 
requirements are the Einstein gravity and the $f(R)$ gravity, the latter of 
which is mathematically equivalent to the Einstein gravity coupled with 
a scalar field with a non-trivial potential.

When we extend general relativity to higher dimensions, however, 
the difference between the two versions of the third requirement becomes 
important. In fact, we require the stronger version, we obtain the same field 
equations for gravity in higher dimensions. In contrast, we require 
the weaker version, we obtain a larger class of theories that contains 
the higher-dimensional general relativity as a special case. 
Such extensions to higher dimensions were first studied systematically 
by D. Lovelock in 1971\cite{Lovelock.D1971}. What he found was that 
the second-rank gravitational tensor $E_{MN}$ balancing the 
energy-momentum tensor $T_{MN}$ is a sum of polynomials in the curvature 
tensor and that the polynomial of each degree is unique up to 
a proportionality constant. The physical importance of such an extension was 
later recognised when B. Zwiebach\cite{Zwiebach.B1985} 
pointed out that the special quadratic combinations of the curvature tensor 
named the Gauss-Bonnet term naturally arises when we add quadratic terms of 
the Ricci curvature to the quadratic term in the Riemann curvature tensor obtained as the $\alpha'$ correction to 
the field equations in the heterotic string theory to obtain a ghost-free 
theory. In this section, we briefly overview 
the Lovelock theory, its static black hole solution and perturbation theory 
recently developed 
by Takahashi and Soda\cite{Takahashi&Soda09a,Takahashi.T&Soda2010A}.

For notational convenience, we introduce an orthonormal basis on $\M$, 
and throughout this section, we use upper case latin indices in the range 
$A,B,\dots,J$ to label the $1$-form $\theta^A$, ($A=0,\cdots,D-1$), 
of the basis, while we use upper case latin indices in the range 
$K,L,M,N,\dots$ to denote tensors on $\M$ as in the rest of this chapter.  

\subsection{Lovelock theory}  

In 1986, B. Zumino\cite{Zumino1986} pointed out that the class of theories obtained by Lovelock 
has a natural mathematical meaning. That is, the Lovelock equations 
for gravity can be obtained from the action that is a linear combination of 
the terms each of which corresponds to the Euler form in some even dimensions:
\Eqr{
S &=& \int \sum_{k=0}^{[(D-1)/2]} \alpha_k \LL_k;\\
\LL_k &:=& \frac{1}{(D-2k)!}\epsilon_{A_1\cdots A_{2k} B_1\cdots B_{D-2k}} 
\R^{A_1A_2}\w
 \cdots\w \R^{A_{2k-1}A_{2k}}\w \theta^{B_1\cdots B_{D-2k}}
\notag\\
&=&  I_k \Omega_D \,,
}
where $\R^A{}_B$ is the curvature form with respect to the orthonormal 
$1$-form basis $\theta^A$ ($A,B=0,\cdots,D-1$), 
$\theta^{A\cdots B}=\theta^A\w\cdots \w \theta^B$, 
$\Omega_D$ is the volume form, and  
\Eq{
I_k= \frac{1}{2^k}\delta^{A_1\cdots A_{2k}}_{B_1\cdots B_{2k}}
  R_{A_1A_2}{}^{B_1B_2}\cdots R_{A_{2k-1}A_{2k}}{}^{B_{2k-1}B_{2k}} \,.
}
The explicit expressions for small values of $k$ are
\Eqrsub{
I_0 &=& 1 \,,\\
I_1 &=& R \,,\\
I_2 &=& R^2 - 4 R^A_B R^B_A + R_{AB}{}^{CD} R_{CD}{}^{AB} \,,\\
I_3 &=& R^3 - 3 R(- 4 R^A_B R^B_A + R_{AB}{}^{CD} R_{CD}{}^{AB}) 
 + 24 R^A_B R^B_C R^C_A + 3 R_{AB}{}^{CD}R_{CD}{}^{EF}R_{EF}{}^{AB}\,. 
\notag\\
&&
}
Hence, if we require that the theory has the Einstein theory in the low energy limit, $\alpha_0$ and $\alpha_1$ are related to the cosmological constant $\Lambda$ and the Newton constant $\kappa^2=8\pi G$ as 
\Eq{
\alpha_0=-\frac{\Lambda}{\kappa^2},\quad
\alpha_1=\frac{1}{2\kappa^2}.
}
Further, if the Gauss-Bonnet term $ \LL_2$ comes from 
the $\Order{\alpha'}$ correction in the heterotic string theory, $\alpha_2>0$.

From the Bianchi identity $\D\R_{AB}\equiv 0$, the variation of the Lagrangian density can be written
\Eqr{
\delta \LL_k 
 &=& d(*)
  + \frac{k}{(D-2k-1)!}  \epsilon_{A_1\cdots A_{2k}B_1\cdots B_{D-2k}}
 \delta\omega^{A_1 A_2} \w \R^{A_3A_4}\w \cdots \Theta^{B_1}\w\theta^{B_2\cdots B_{D-2k}}
\notag\\
&& +
\frac{(-1)^{D-1}}{(D-2k-1)!}\epsilon_{A_1\cdots A_{2k}B_1\cdots B_{D-2k}}
 \R^{A_1A_2}\w\cdots \R^{A_{2k-1}A_{2k}}\w\theta^{B_1\cdots B_{D-2k-1}}\w
\delta\theta^{B_{D-2k}}
\notag\\
 &=& d(*)  + \inrbra{k T^{(k)}{}^A_{B_1 B_2} (\delta\omega^{B_1B_2})_A
    + (-1)^{D-1} E^{(k)}{}^A_B \delta\theta^B_M e^M_A }\Omega_D \,,
}
where $\omega^A{}_B$ is the connection form with respect to $\theta^A$, $e_A$ is the vector basis dual to $\theta^A$, $\D$ is the corresponding covariant exterior derivative, $\Theta^A=\D\theta^A$ is the torsion 2-form, and 
\Eqrsub{
T^{(k)}{}^A_{B_1B_2} &=& \delta^{A C_1\cdots C_{2k}}_{B_1B_2 D_1\cdots D_{2k-1}}
 R^{D_1D_2}{}_{C_1C_2}\cdots R^{D_{2k-3}D_{2k-2}}{}_{C_{2k-3}C_{2k-2}}
 T^{D_{2k-1}}_{C_{2k-1}C_{2k}} \,,
\label{Lovelock:GFEQ:torsion}\\
E^{(k)}{}^A_B 
 &=& -\frac{1}{2^k} \delta^{A A_1\cdots A_{2k}}_{B B_1\cdots B_{2k}}
   R_{A_1A_2}{}^{B_1B_2}\cdots  R_{A_{2k-1} A_{2k}}{}^{B_{2k-1}B_{2k}} 
\,.
\label{Lovelock:GFEQ:curvature}
}
Hence, if we treat $\theta^A$ and $\omega^A{}_B$ 
as independent dynamical variables, the field equations are given by
\Eqrsub{
&& \sum_k \alpha_k E^{(k)}{}^A_B=0 \,,\\
&& \sum_k k\alpha_k T^{(k)}{}^A_{BC}=0 \,. 
}
If we require the connection to be Riemannian, the second equation becomes 
trivial due to the torsion free condition $T^A_{BC}=0$. 
Examples of the explicit expressions for $E^{(k)}{}^A_B$ and 
$T^{(k)}{}^A_{BC}$ are 
\Eqrsub{
E^{(0)}{}^A_B &=& -\delta^A_B \,,\\
E^{(1)}{}^A_B &=& 2 R^A_B - R \delta^A_B \,,\\
E^{(2)}{}^A_B &=& -\delta^A_B I_2 + 4RR^A_B-8 R^{AC}{}_{BD}R^D_C 
\nonumber \\
      &{}&\; \quad \qquad    
      + 8 R^A_C R^C_B + 4 R^{A C_1C_2C_3}R_{BC_1C_2C_3} \,,\\
T^{(1)}{}^A_{BC} &=& T^A_{BC}+ 2 \delta^A_{[B} T^D_{C]D} \,,\\
T^{(2)}{}^A_{BC} &=& 8\delta^A_{[B}(-2 R^D_{C]}T_D + R^{D_1D_2}{}_{C]D_3}T^{D_3}_{D_1D_2}
        +RT_{C]}-2 R^{D_1}_{D_2}T^{D_2}_{C]D_1} )
\notag\\
 && + 12 (R^{AD}{}_{BC}T_D-2 R^A_{[B}T_{C]} - 2 R^{AD_1}{}_{D_2[B}T^{D_2}_{C]D_1}
    -12 R^A_D T^D_{BC}) \,.
}
%

\subsection{Static black hole solution}

\subsubsection{Constant curvature spacetimes}

In general relativity, a constant curvature spacetime is always a vacuum solution and the curvature is uniquely determined by the value of the cosmological constant. This feature is not shared by the Lovelock theory. In fact, the Lovelock theory does not allow a vacuum constant curvature solution for some range of the coupling constants $\Set{\alpha_k}$, and have multiple constant curvature solutions with different curvatures for other range of the coupling constants. 

To see this, let us insert the constant Riemann curvature
\Eq{
R_{MNLK}= \lambda (g_{ML}g_{NK}-g_{MK}g_{NL})
}
into the field equation \eqref{Lovelock:GFEQ:curvature}. Then, we obtain
\Eq{
P(\lambda)=0 \,,
}
where
\Eq{
P(X):=\sum_{k=0}^{[(D-1)/2]} \alpha_k \frac{X^k}{(D-2k-1)!} \,.
}
For general relativity for which $\alpha_k=0$ for $k\ge2$, this equation has a unique solution. In contrast, when $D>4$ and $\alpha_k\neq$ ($k\ge 2$), the equation can have no solution or multiple solutions depending on the functional shape of $P(X)$. 

\subsubsection{Black hole solution}

Now, let us look for spherically symmetric black hole solutions. 
Because the Birkhoff-type theorem holds for the Lovelock theory except for 
the case in which $P(\lambda)=0$ has a root with multiplicity higher than 
one\cite{Zegers.R2005,Wiltshire.D1986}, 
we only consider static spacetimes whose metric can be put into the form
\Eq{
ds^2=  -f(r) dt^2 + \frac{dr^2}{h(r)} + r^2 d\sigma_n^2 \,,
}
where $d\sigma_n^2$ represents the metric of a constant curvature space 
with sectional curvature $K$. For a spherically symmetric solution, $K=1$. 
However, because the argument in this section holds for any value of $K$, 
we consider this slightly general spacetime. 

Then, the non-vanishing components of the curvature tensor are given up to 
symmetry by 
\Eqrsubl{SSST:curvature}{
&& R^{01}{}_{01}=\frac{h}{2}\inpare{ -\frac{f''}{f} + \frac{(f')^2}{2f^2}}
 -\frac{h'f'}{4f} \,, \\
&& R_{0i0j}=\frac{hf'}{2rf} g_{ij} \,,\quad
   R_{1i1j}=-\frac{h'}{2r} g_{ij} \,, \\
&& R_{ijkl}=X (g_{ik}g_{jl}-g_{il}g_{jk}) \,,
}
where $X(r):=(K-h(r))/r^2$. Inserting these into 
Eq.~\eqref{Lovelock:GFEQ:curvature}, 
we find that the field equations reduce to 
\Eq{
(r^{n+1}P(X(r)))'=0 \,,\quad
P^{(1)}(X(r))(f(r)/h(r))'=0 \,. 
}
The first of these is integrated to yield
\Eq{
P(X(r))= \frac{C}{r^{n+1}} \,, 
}
where $C$ is an integration constant. This determines the function $h(r)$ 
implicitly. In particular, for $C=0$, we have $h(r)=K-\lambda r^2$ 
for each solution to $P(\lambda)=0$. If $P^{(1)}(\lambda)\neq 0$, after 
an appropriate scaling of $t$, $f(r)=h(r)$ follows from the second of 
the above field equations: 
\Eq{
ds^2=-(K-\lambda r^2) dt^2 + \frac{dr^2}{K-\lambda r^2} + r^2 d\sigma_n^2 \,. 
}
This represents a constant curvature spacetime with sectional curvature 
$\lambda$ irrespective of the value of $K$, as is well known. 

This implies that for $C\neq0$, we have in general multiple solutions 
corresponding to multiple solutions to $P(\lambda)=0$. Each solution 
approaches a constant curvature spacetime with sectional curvature $\lambda$ 
at large $r$ asymptotically. For these solutions, we can always 
put $f(r)=h(r)$ and the metric can be written\cite{Wheeler.JT1986}  
\Eq{
ds^2= -f(r) dt^2 + \frac{dr^2}{f(r)}+ r^2d\sigma_n^2,\quad
f(r)=K-X(r) r^2.
}
In general, the constant $C$ is proportional to the total mass $M$ of the system and positive if $M>0$. We can easily show that $X(r)$ changes monotonically with $r$ from infinity to some value of $r$ where the metric becomes singular. This singularity may or may not be hidden by a horizon depending on the functional shape of $P(X)$. In the former case, we obtain a regular black hole solution. 

\subsection{Perturbation equations for the static solution}  

The linear perturbation of the Lovelock tensor, 
Eq.~\eqref{Lovelock:GFEQ:curvature}, in general reads 
\Eq{
\delta E^M_N = -\sum_{k=1}^{[(D-1)/2]} \frac{k\alpha_k}{2^k}\delta^{M M_1\cdots M_{2k}}_{N N_1\cdots N_{2k}} R_{M_1 M_2}{}^{N_1N_2}\cdots R_{M_{2k-3}M_{2k-2}}{}^{N_{2k-3}N_{2k-2}}
 \delta R_{M_{2k-1}M_{2k}}{}^{N_{2k-1}N_{2k}} \,.
}
Inserting Eq.~\eqref{SSST:curvature} with $f(r)=h(r)=K-X(r)r^2$ into this 
yields\cite{Takahashi.T&Soda2010A}
\Eqrsubl{Lovelock:PGEQ}{
r^{n-1} \delta E^t_t  
  &=& -\frac{r T'}{n-1}\delta R_{ij}{}^{ij} - 2 T \delta R_{i r}{}^{i r},
  \label{Lovelock:PGEQ:tt}\\
r^{n-1} \delta E^r_t 
  &=& -2 T \delta R_{it}{}^{ir},
  \label{Lovelock:PGEQ:tr}\\
r^{n-1} \delta E^r_r 
  &=& - \frac{r T'}{n-1}\delta R_{ij}{}^{ij} - 2T \delta R_{it}{}^{it},
  \label{Lovelock:PGEQ:rr}\\
r^{n-1} \delta E^i_a 
  &=& \frac{2rT'}{n-1}\delta R_{a j}{}^{ij} + 2T \delta R_{a b}{}^{ib},
  \label{Lovelock:PGEQ:ai}\\
r^{n-1} \delta E^i_j 
  &=& \frac{2rT'}{n-1} \delta R_{aj}{}^{ai}
  +\frac{2r^2 T''}{(n-1)(n-2)} \delta R_{jk}{}^{ik}
  \notag\\
  &&-\delta_j^i \insbra{2T\delta R_{tr}{}^{tr}+\frac{2rT'}{n-1}\delta R_{ak}{}^{ak}
    + \frac{r^2 T''}{(n-1)(n-2)} \delta R_{kl}{}^{kl}},
  \label{Lovelock:PGEQ:ij}
}
where
\Eq{
T(r):= r^{n-1} P^{(1)}(X(r)).
}
Thus, we can obtain perturbation equations for the metric in the Lovelock theory simply by calculating the perturbation of the curvature tensor.

\subsubsection{Tensor perturbations}

For tensor perturbations, the metric perturbation can be expanded in terms of 
the harmonic tensor $\THB_{ij}$ as Eq.~\eqref{TensorPerturbation:metric}. 
The non-vanishing components of the curvature tensor for this type of 
perturbations read 
\Eqrsub{
\delta R_{ai}{}^{aj} &=& -\inpare{\Box H_T + \frac{2}{r}Dr\cdot D H_T} \THB_i^j,\\
\delta R_{ik}{}^{jk} &=& \insbra{-(n-2)\frac{f'}{r}H_T' + \frac{2K+k_t^2}{r^2} H_T}\THB_i^j.
} 
Hence, from the above expression for $\delta E^j_i$, we obtain the following wave equation for $H_T$:
\Eq{
\frac{1}{f}\ddot H_T- f H_T''-\inpare{ f\frac{T''}{T'}+\frac{2f}{r}+f'}H_T'
  + \frac{2K+k_t^2}{(n-2)r} \frac{T''}{T'} H_T=0.
}
If we introduce the mode function $\Psi(r)$ by 
\Eq{
H_T(t,r)= \frac{\Psi(r)}{r\sqrt{T'(r)}} e^{-i\omega t},
}
this wave equation can be put into the standard form
\Eq{
-\frac{d^2\Psi}{dr_*^2} + V_t\Psi =\omega^2\Psi,
}
with the effective potential
\Eq{
V_t(r)=\frac{(2K+k_t^2)f}{(n-2)r} \frac{T''}{T'} + \frac{1}{r\sqrt{T'}}\frac{d^2 (r\sqrt{T'})}{dr_*^2},
}
where $dr_*=dr/f(r)$ as in the Einstein black hole case.

\subsubsection{Vector perturbations}

For vector perturbations, the perturbation of components of the curvature tensor that are relevant to the field equations can be expressed in terms of the basic gauge-invariant quantities $F_a$ as
\Eqrsub{
\delta R_{aj}{}^{ai} &=& -\frac{k}{r^2} D^a (r F_a) \VHB^i_j,\\
\delta R_{jk}{}^{ik} &=& -\frac{(n-2)k}{r} \frac{D^a r}{r} F_a \VHB^i_j,\\
\delta R_{aj}{}^{ij} &=& \frac{n-1}{2r^2}\insbra{-D^b F^{(1)}_{ba}
  - \frac{k^2-(n-1)K}{(n-1)r} F_a 
   + \frac{2(K-f)+rf'}{r} f_a } \VHB^i,\\
\delta R_{ab}{}^{ib} &=& \insbra{-\frac{1}{2r^3}D^b(r^2 F^{(1)}_{ba})
   +\frac{rf''-f'}{2r^2} f_a } \VHB^i.
}
Inserting these into Eq.~\eqref{Lovelock:PGEQ:ai}, we obtain the following 
equations for the gauge-invariant variables: 
\Eqrsub{
&& \frac{1}{r^2}D^b \inpare{r^2 T F^{(1)}_{ba}}
         + T'\frac{k^2-(n-1)K}{(n-1)r} F_a =0,\\
&& \frac{1}{r} D^a (rT' F_a)=0.
}
The gauge-dependent residuals in the expressions for the curvature 
tensor cancel exactly owing to the identity 
\Eq{
P^{(1)} X'' + P^{(2)} (X')^2 + \frac{n+2}{r} P^{(1)}X'=0
}
obtained from the background equation $(r^{n+1}P(X))'=0$. 
We can easily confirm that for general relativity for which 
$T=r^{n-1}/(2\kappa^2)$, these reduce to 
Eqs.~\eqref{eq:Vectorperturbation:evol:Fa} and 
\eqref{eq:Vectorperturbation:constr:Fa} with no source terms.

A master equation for vector perturbations in the Lovelock theory can be 
derived in the same way as that in general relativity. First, the second 
perturbation equation implies the existence of a potential $\Omega$ in which 
$F_a$ can be expressed as 
\Eq{
rT' F_a =\epsilon_{ab}D^b \Omega.
}
Inserting this into the first perturbation equation, we easily find that 
it is equivalent to 
\Eq{
rT D_a\inpare{\frac{1}{r^2 T'} D^a\Omega} 
  - \frac{k^2- (n-1)K}{(n-1)r^2} \Omega=0.
}
%

\subsubsection{Scalar perturbations}

For scalar perturbations, we have
\Eqrsub{
\delta R^{ai}{}_{bi} &=& \frac{k^2}{2r^2} F^a_b + \frac{nf}{r}D_{[b}F^a_{r]}
 + \frac{n}{2r}D^a F^r_b + \frac{nf'}{2r} F^a_b 
 -n D_b D^a F 
\notag\\
 && -\frac{n}{r}(D^a r D_b + D_b r D^a)F
  + \frac{n}{2}X^c D_c \pfrac{f'}{r} \delta^a_b,\\
\delta R^{ai}{}_{aj} &=& -\frac{k^2}{2r^2}F^a_a \SHB^i_j
   + \left\{ \frac{D_a r}{r}D_b F^{ab} -\frac{1}{2r}Dr\cdot D F^a_a 
 \right.
 \notag\\
 && \left.\quad
  + \inpare{\frac{f'}{2r}+ \frac{k^2}{2nr^2} } F^a_a
  -\frac{1}{r^2}D^a (r^2 D_a F) 
  + X^a D_a \pfrac{f'}{r}
  \right\} \delta^i_j \SHB
 \\
\delta R^{ik}{}_{jk} &=& (n-1) \left[ 
  \frac{D^a r D^b r}{r^2} F_{ab} -\frac{2}{r}Dr^a D_a F
  + \frac{2(k^2-nK)}{nr^2}F
\right. \notag\\
 &&\left. \qquad
 - D^b \pfrac{ K-f}{r^2} X_b \right] \delta^i_j\SHB
 -(n-2)\frac{k^2}{r^2}F \SHB^i_j,\\
\delta R^{ib}{}_{ab} &=& \insbra{
 -\frac{k}{r} D_{[b} \inpare{\frac{1}{r} F^b_{a]} }
   + \frac{1}{2kr} (f'-rf'') D_a H_T } \SHB^i,\\
\delta R^{ij}{}_{aj} &=& (n-1)k\insbra{ -\frac{D_b r}{2r^3}F^b_a
 + \frac{1}{r^2} D_a F
 -\frac{2(K-f)+ rf'}{2r^2k^2} D_a H_T } \SHB^i,\\
\delta R^{ab}{}_{ab} &=& \inrbra{ 
 -\Box F^a_a + D^a D^b F_{ab} + \frac{f''}{2}F^a_a 
   + X^a D_a f'' } \SHB \,. 
}
Inserting these into Eq.~\eqref{Lovelock:PGEQ}, we obtain
\Eqrsub{
r^{n-1}\delta E^t_t &\equiv& \left[-\frac{nfT}{r}(F^r_r)'
 -\inpare{\frac{nf'}{r}T+\frac{nT'}{r}f+\frac{k^2}{r^2}T} F^r_r 
 +2nfT F'' 
 \right.
 \notag\\
&& \left.
 + \inpare{\frac{4nf}{r}T+2nfT'+nf'T}F'
 -2T' \frac{k^2-nK}{r}F
 \right] \SHB=0 \,,\\
r^{n-1}\delta E^r_t &\equiv& \left[-\frac{nfT}{r}\insbra{ \d F^r_r+ \frac{k^2}{nrf}F^r_t
-2\sqrt{f} \inpare{\frac{r}{\sqrt{f}}\d F}'} \right]\SHB =0 \,, \\
r^{n-1}\delta E^r_r &\equiv& \left[ \frac{Tnf}{r}\inrbra{(F^t_t)'
  -\frac{k^2}{nfr}F^t_t } 
  -\frac{2nfT}{r}\d F^t_r - \frac{nfT}{r}\inpare{\frac{f'}{f}+ \frac{T'}{T}} F^r_r
\right.\notag\\
&& \left.\quad 
-\frac{2nT}{f}\ddot F + nfT \inpare{\frac{f'}{f}+2\frac{T'}{T}} F'
 -2T'\frac{k^2-nK}{r}F
 \right]\SHB=0 \,,\\
r^{n-1}\delta E^i_a &\equiv& \frac{k}{r}\insbra{
 -D_b\inpare{\frac{T}{r}F^b_a} + TD_a\inpare{\frac{1}{r}F^b_b}
  + 2T' D_a F}\SHB^i=0 \,, \\
r^{n-1}\delta E^i_j &\equiv& -\frac{k^2}{n-1}\inpare{\frac{T'}{r} F^a_a + 2 T'' F} \SHB^i_j
\notag\\
 && -\left[ -D^a (TD_a F^b_b)+ \inpare{\frac{(f'T)'}{2} + \frac{k^2 T'}{nr}}F^a_a
  + D^a D^b (TF_{ab}) 
\right. \notag\\
  &&\left. \quad
  -\frac{2}{r}D^a\inpare{ r^2 T' D_a F} + \frac{2(k^2-nK)}{n} T'' F 
\right] \delta^i_j \SHB=0 \,. 
}

As was first shown by Takahashi and Soda\cite{Takahashi.T&Soda2010A}, we can reduce these equations to a single master equation in terms of the master variable $\Phi$ defined by
\Eq{
F^r_t= r(\d\Phi +2\d F),
}
as 
\Eq{
\ddot\Phi - \frac{f A^2}{r^2T'}\inpare{\frac{r^2 f T'}{A^2} \Phi'}' + Q\Phi=0,
}
where
\Eqrsub{
&& A=2k^2 -2nf + nr f',\\
&& Q=\frac{f}{nr^2T} \insbra{r(k^2T+nrfT')\inpare{ 2\frac{(AT)'}{AT}-\frac{T''}{T'}}
 -n(r^2 f T')'}.
}
The other gauge-invariant variables are expressed in terms of $\Phi$ as
\Eqrsub{
&& F=-\frac{1}{A}\inrbra{ nrf\Phi' + \inpare{k^2+ nrf\frac{T'}{T}}\Phi},\\
&& F^r_r = -\frac{k^2}{nf}\Phi + 2rF' -\frac{A}{nf}F,\\
&& F^t_t= -F^r_r -\frac{2rT''}{T'} F.
}
%

\section{Stability Analysis} 
\label{sec:stabilityanalysis}

\subsection{Stability criterion and $S$-deformation}

With the decoupled master equations in hand, we are 
ready to study the stability of generalized static black holes with 
charge and cosmological constant in Einstein-Maxwell theory, 
whose metric is given by Eqs.~\eqref{f:RNBH} and \eqref{metric:GSBH}.  
We consider only the stability in the static region outside the black 
horizon. This region is represented as $r>r_H$ for $\lambda\le0$ 
and $r_H<r<r_c$ for $\lambda>0$. 
Such a region exists only for restricted ranges of 
the parameters $M,Q$ and $\lambda$. 
[See Appendix A of Paper~III for details.] 

We have seen before that for any perturbation type, the master equations 
for perturbation in the static region are reduced to an eigenvalue 
problem of the type
\Eq{
\omega^2 \Phi = A\Phi \,, 
}
where $A$ is the derivative operator 
\Eq{
A = -\frac{d^2}{dr_*^2} + V(r); \quad dr_*=\frac{dr}{f} \,,
}
with $V(r)$ being equal to $V_T(r)$, $V_{V\pm}(r)$ or $V_{S\pm}(r)$.   
The operator $A$ is self-adjoint (by imposing suitable boundary conditions 
if necessary) in the standard square integrable function space 
$L^2(r_*,dr_*)$. Therefore, if its spectrum is non-negative, 
there is no exponentially growing mode among physically acceptable 
(i.e., normalizable) modes, implying the stability of the black hole.

When $\lambda \geq 0$, the static region where $A$ is defined 
is globally hyperbolic and the range of $r_*$ is complete. 
In this case the operator $A$ is essentially self-adjoint 
with the domain of smooth functions of compact support, denoted hereafter 
by $C^\infty_0(r_*)$, and has the unique self-adjoint extension called 
the Friedrichs extension $A_F$, which is given by taking the closure of 
$(A,C^\infty_0(r_*))$ and is known to have the same lower bound 
of the spectrum of $A$ with domain $C^\infty_0(r_*)$.  

%
However, when $\lambda <0$, $r_*$ has an upper bound, and whether $A$ 
becomes essentially self-adjoint depends upon the asymptotic behavior of 
the potential $V(r)$ in $A$, and thus upon the type of perturbations 
as well as the spacetime dimension. 
When $A$ is not essentially self-adjoint, there are infinitely many 
different choices of boundary conditions that make $A$ self-adjoint, 
and the spectrum of a self-adjoint extension depends upon 
the associated boundary conditions at the upper-bound of 
$r_*$ ($r\rightarrow \infty$). 
In particular, even if $A$ with $C^\infty_0(r_*)$ is positive-definite, 
its extension can admit a negative spectrum, depending on the choice of 
boundary conditions. 
This could be the case for the vector-type of electromagnetic 
and gravitational perturbations in $n=2$ and for the scalar-type of 
electromagnetic and gravitational perturbations in $n=2,3,4$, 
as shown in a simple, massless case by inspecting the asymptotic 
behavior of $A$\cite{Ishibashi.A&Wald2004}. 
For these cases, we must specify a boundary condition for $\Phi$ at 
$r=\infty$, 
and in the following we simply adopt the setting zero Dirichlet condition, 
$\Phi\tend 0$ as $r \tend \infty$, which corresponds to the Friedrichs 
extension, $A_F$. 

Now suppose $\Phi(r)$ is a smooth function of compact support contained in
$r>r_H$ (or $r_H<r<r_c$ for $\lambda>0$). Then, we can 
rewrite the expectation value of $A$, $(\Phi,A\Phi)$, as
\Eq{
(\Phi,A\Phi)=\int dr_*\left( \left|\frac{d\Phi}{dr_*}\right|^2 
  +V|\Phi|^2 \right).
\label{EVofA}
}
Note that no boundary terms appear as $\Phi \in C^\infty_0(r_*)$. 
When the static region is globally hyperbolic, we have 
the Friedrichs extension $A_F$. Then, in order to show the stability, 
it is sufficient to show the positivity of the right-hand side 
of Eq.~\eqref{EVofA} 
for $\Phi \in C^\infty_0(r_*)$\cite{Ishibashi.A&Kodama2003A}, 
since $A_F$ has the same lower bound of $(A,C^\infty_0(r_*))$.    
%
In particular, this condition obviously is satisfied if $V$ is non-negative. 
For $4$-dimensional Schwarzschild black hole, this is indeed the case. 
However, in higher dimensions, $V$ is in general not positive definite,  
and it is far from obvious to know whether $A$ is positive. 
One powerful method to show the positivity of $A$ beyond such a simple 
situation is the procedure called the {\em $S$-deformation} of $V$ 
in Papers~II and III, in which we deform the right-hand side of 
Eq.~\eqref{EVofA} by partial integration in terms of a function $S$ as 
\Eq{
(\Phi,A\Phi)= \int dr_* \left( |\tilde D\Phi|^2+\tilde V|\Phi^2| 
\right) \,, 
\label{EVofA:S}
}
where 
\Eqr{
 \tilde D=\frac{d}{dr_*}+S \,,\quad 
 \tilde V= V + f\frac{dS}{dr} -S^2 \,, 
}
and $\Phi \in C^\infty_0(r_*)$. Our task is now to find a suitable 
function $S$ which makes the effective potential ${\tilde V}$ positive. 
In the following we quote the results of 
the stability analysis of Papers~II and III. 

One might worry that our boundary conditions $\Phi \rightarrow 0$ 
in a neighbourhood of the horizon $r\rightarrow r_H$ would be too strong. 
However, if the horizon is non-degenerate and accordingly admits a bifurcate 
surface and once the stability is shown under our boundary conditions,  
then we can conclude by applying the theorem of 
Kay and Wald\cite{Kay.B&Wald87} that the black hole spacetime is stable 
under perturbations that are non-vanishing at the bifurcate surface. 
Below we shall examine the stability for each type of perturbations. 

\subsection{Tensor perturbation} 
The stability analysis of higher dimensional static vacuum black holes
under tensor type perturbations has first been examined by Gibbons and 
Hartnoll\cite{Gibbons.G&Hartnoll2002} with special interests in the case
where $\K^n$ is a generic Einstein manifold. In the following we review
our analysis in Paper~III, which generalize their
results\cite{Gibbons.G&Hartnoll2002}. 

Consider the potential, Eq.~\eqref{VT:GSBH}. As discussed in Paper~II, 
with the choice 
\Eq{
S=-\frac{nf}{2r} \,,  
}
the $S$-deformation yields 
\Eq{
\tilde V_T =\frac{f}{r^2}\left[\lambda_L-2(n-1)K \right] \,, 
}
irrespective of the $r$-dependence of $f(r)$. 
Therefore $\tilde V_T$ becomes positive if 
\Eq{
\lambda_L\ge 2(n-1)K \,.
\label{StabilityCondition:Tensor}
}
In particular, this immediately guarantees the stability of maximally 
symmetric black holes for $K=1$ and $K=0$, since $\lambda_L$ is related to 
the eigenvalue $k_T^2$ of the positive operator $- {\hat D}^i {\hat D}_i$ as 
$\lambda_L=k_T^2+2nK$ when $\K^n$ is maximally symmetric. 
As for the case $K=-1$, $\tilde V_T$ could be negative 
even in the maximally symmetric case, if $0<k_T^2<2$. This is 
however not possible, due to the bound, Eq.~\eqref{value:kt2}. 
Therefore we conclude that the maximally symmetric black holes 
considered here are stable under tensor type perturbations.   

Note that the condition, Eq.~\eqref{StabilityCondition:Tensor}, is just a 
sufficient condition for stability, and it is not a necessary 
condition in general. In fact, for tensor type perturbation, we can 
obtain stronger stability conditions directly from the positivity of 
$V_T$ if we restrict the range of parameters. For example, for $K=1$ 
and $\lambda=0$, it is easy to see that $V_T$ is positive if
\Eq{
\lambda_L+2-2n+\frac{n(n-1)\sqrt{M^2-Q^2}}{M+\sqrt{M^2-Q^2}}\ge 0
}
for $M^2\ge Q^2>8n(n-1)M^2/(3n-2)^2$, and 
\Eq{
\lambda_L+\frac{n^2-10n+8}{4}\ge 0
\label{StabilityCond:Tensor:smallQ}}
for $Q^2\le 8n(n-1)M^2/(3n-2)^2$. Thus, if we do not restrict the 
range of $Q^2$, we obtain the same sufficient condition for 
stability as Eq.~\eqref{StabilityCondition:Tensor}, but for the 
restricted range $Q^2\le 8n(n-1)M^2/(3n-2)^2$, we obtain the 
stronger sufficient condition, Eq.~\eqref{StabilityCond:Tensor:smallQ}, 
which coincides with the condition obtained in Paper~II for the case 
$Q=0$.


Similarly, for $K=-1$ and $\lambda<0$, 
%
if we restrict the range of $\lambda$ to
\Eq{
\lambda\ge -\pfrac{(n+1)M}{nQ^2}^{\frac{2}{n-1}}
  \left( 1+\frac{(n^2-1)M^2}{n^2Q^2} \right) \,, 
}
we obtain from $V_T>0$ a sufficient condition for stability stronger than 
Eq.~\eqref{StabilityCondition:Tensor}, 
\Eq{
\lambda_L+3n-2=k_T^2+n-2\ge0 \,. 
}
This condition is sufficient to guarantee the stability of a 
maximally symmetric black hole with $K=-1$ for $n\ge2$. 
However, if we extend the range of $\lambda$ to the whole allowed range, 
then Eq.~\eqref{StabilityCondition:Tensor} becomes the strongest condition 
that can be obtained only from $V_T>0$.

\subsection{Vector perturbation}

For the choice of 
\Eq{
S=\frac{nf}{2r}, 
}
$V=V_{V\pm}$ in Eq.~\eqref{Vpm:Vector} are deformed to 
\Eqr{
&& \tilde V_{V\pm} =\frac{f}{r^2}\left[ m_V
   +\frac{(n^2-1)M\pm\Delta}{r^{n-1}} \right],\\
&& m_V=k_V^2-(n-1)K.
}
It follows from $m_V \geq 0$ that $\tilde V_{V+}$ is positive definite,
and therefore static charged black holes are stable under the
electromagnetic mode of vector type perturbation. 

For the gravitational mode, we have 
\Eqr{
&& \tilde V_{V-}=\frac{f}{r^2}
       \frac{m_V h}{(n^2-1)M+\Delta} \,,\\
&& h:=(n^2-1)M-\frac{2n(n-1)Q^2}{r^{n-1}}+\Delta \,. 
}
Since $h$ is a monotonically increasing function of $r$, $\tilde V_{V-}$ is 
positive if and only if $h(r_H)\ge0$. Therefore $\tilde V_{V-}$ may
become negative. 
In the case of $\lambda\ge0$, the background 
spacetime contains a regular black hole only for $K=1$, and the 
static region outside the black hole is given by 
$r_H<r<r_c$ $(\le+\infty)$. In this region, it turns out that 
$h>0$ and therefore the black hole is stable.
In the case of $\lambda<0$, 
under the condition that the spacetime contains a regular black hole,  
%
we have the relation 
\Eq{
h\ge \sqrt{(n^2-1)^2M^2+2n(n-1)m_V Q^2}
    -\sqrt{(n^2-1)^2M^2-4Kn(n-1)^2Q^2}.
}
It follows that for $K=0,1$, $h>0$. 
As for $K=-1$, the right-hand side of this inequality could 
become negative if $k_v^2<n-1$. This is however not possible 
as shown just below Eq.~\eqref{ids:Vector}. Therefore, $h>0$ also for $K=-1$. 
We conclude that the black holes considered here are stable under 
vector type perturbations. 

\begin{figure}[t]
\centerline{\includegraphics[width=5cm]{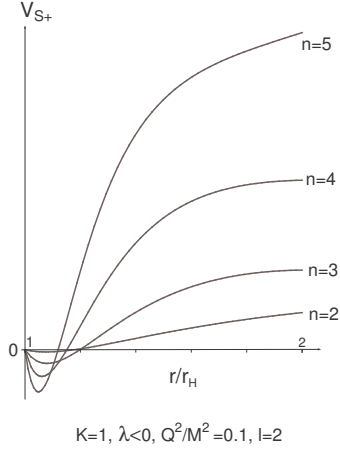}}
\caption{Examples of $V_{S+}$ for $K=1$ and $\lambda<0$.}
\label{fig:VSplus}
\end{figure}

\subsection{Scalar perturbation}

By applying the $S$-deformation to $V_{S+}$ with 
\Eq{
S=\frac{f}{h_+}\frac{dh_+}{dr}\,, \quad 
h_+=r^{n/2-1}H_+ \,, 
}
we obtain 
\Eq{
\tilde V_{S+}=\frac{k^2 f }{2r^2 H_+} \left[ (n-2)(n+1)\delta x + 2\right] \,.
}
Since this is positive definite, the electromagnetic mode $\Phi_+$ 
is always stable for any values of $K$, $M$, $Q$ and $\lambda$, 
provided that the spacetime contains a regular black hole, although 
$V_{S+}$ has a negative region near the horizon when $\lambda<0$ and 
$Q^2/M^2$ is small (see Fig. \ref{fig:VSplus}).

Using a similar transformation, we can also prove the stability with
respect to the gravitational mode $\Phi_-$ for some special cases.
For example, the $S$-deformation of $V_{S-}$ with 
\Eq{
  S=\frac{f}{h_-}\frac{dh_-}{dr} \,, \quad 
  h_-=r^{n/2-1}H_- 
}
leads to
\Eq{
\tilde V_{S-}=\frac{k^2f}{2r^2H_-}
  \left[ 2m-(n+1)(n-2)(1+m\delta)x \right].
}
For $n=2$, this is positive definite for $m>0$. When $K=1$, 
$\lambda\ge0$ and $n=3$ or when $\lambda\ge0,Q=0$ and the horizon is 
$S^4$, 
we can show that $\tilde V_{S-}>0$. Hence, in these special cases, 
the black hole is stable with respect to any type of perturbation. 

However, for the other cases, $\tilde V_{S-}$ is not positive 
definite for generic values of the parameters. The $S$-deformation 
used to prove the stability of neutral black holes in Paper~II is 
not effective either. This is because $V_{S-}$ is negative 
in the immediate vicinity of the horizon for the extremal and 
near extremal cases, as shown in Fig. \ref{fig:VSminus}, and 
the $S$-deformation cannot remove this negative region 
if $S$ is a regular function at the horizon. 
Hence, the stability problem for these generic cases with $n\ge3$ is left open.

\begin{figure}[t]
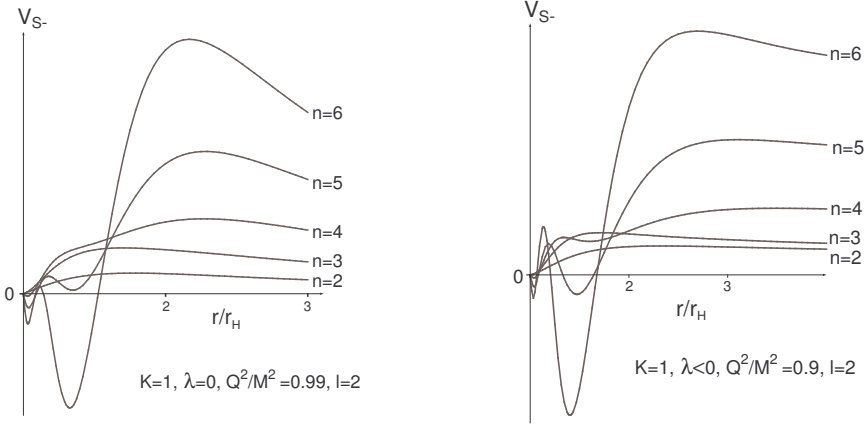

\begin{minipage}{\halftext}
\centerline{\includegraphics[width=6cm]{\FigDir/PotSmK1L0.eps}}
\end{minipage}
\begin{minipage}{\halftext}
\centerline{\includegraphics[width=6cm]{\FigDir/PotSmK1LN.eps}}
\end{minipage}
\caption{Examples of $V_{S-}$.}\label{fig:VSminus}
\end{figure}

Our results are summarized in Table~\ref{tbl:stability}.  
\begin{table}
\caption{Stabilities of generalized static black holes. In this 
table, ``$D$'' represents the spacetime dimension, $n+2$. The 
results for tensor perturbations apply only for maximally symmetric 
black holes, while those for vector and scalar perturbations are 
valid for black holes with generic Einstein horizons, except in the 
case with $K=1,Q=0,\lambda>0$ and $D=6$.}
\label{tbl:stability}
\begin{tabular}{|l|l|c|c|c|c|c|c|}
\hline\hline

\multicolumn{2}{|c|}{}& \multicolumn{2}{c|}{Tensor}
  & \multicolumn{2}{|c|}{Vector}& \multicolumn{2}{c|}{Scalar}\\
\cline{3-8}
\multicolumn{2}{|c|}{}&$Q=0$ & $Q\not=0$ &$Q=0$ & $Q\not=0$ 
&$Q=0$ & $Q\not=0$ \\
\hline
$K=1$& $\lambda=0$ & OK & OK & OK & OK 
     & OK 
     & $\begin{array}{l} 
         D=4,5\ \text{OK} \\ D\ge6\ \text{?} 
       \end{array}$
     \\
\cline{2-8}
     &$\lambda>0$ & OK & OK & OK & OK 
     & $\begin{array}{l} 
         D\le6\ \text{OK} \\ D\ge7\ \text{?} 
        \end{array}$
     & $\begin{array}{l}
         D=4,5\ \text{OK} \\ D\ge6\ \text{?} 
        \end{array}$
     \\
\cline{2-8}
     &$\lambda<0$ & OK & OK & OK & OK 
     &  $\begin{array}{l}
          D=4\ \text{OK} \\ D\ge5\ \text{?} 
         \end{array}$
     &  $\begin{array}{l}
          D=4\ \text{OK} \\ D\ge5\ \text{?} 
         \end{array}$
     \\
\hline
$K=0$ &$\lambda<0$ & OK & OK & OK & OK 
     & $\begin{array}{l}
         D=4\ \text{OK} \\ D\ge5\ \text{?} 
        \end{array}$
     & $\begin{array}{l}
         D=4\ \text{OK} \\ D\ge5\ \text{?} 
       \end{array}$ 
     \\
\hline
$K=-1$ &$\lambda<0$ & OK & OK & OK & OK  
     & $\begin{array}{l}
         D=4\ \text{OK} \\ D\ge5\ \text{?} 
        \end{array}$
     & $\begin{array}{l}
         D=4\ \text{OK} \\ D\ge5\ \text{?} 
        \end{array}$
     \\
\hline
\end{tabular}
\end{table}
As shown there, maximally symmetric black holes are stable with respect to 
tensor and vector perturbations over the entire parameter range.
In contrast, for scalar type perturbation, we were not able to prove 
even the stability of asymptotically flat black holes with charge in 
generic dimensions, due to the existence of a negative region in the 
effective potential around the horizon in the extremal and near 
extremal cases. 

In this regard, it should be noted that by extensive numerical studies, 
Konoplya and Zhidenko\cite{Konoplya&Zhidenko07} have shown that  
Schwarzschild-de Sitter black holes are stable in $D:=2+n =5,...,11$. 
They also found that charged Schwarzschild-de Sitter black holes 
can be unstable if the electric charge and cosmological constant 
are large enough in $D\geq 7$\cite{Konoplya&Zhidenko09ChargeddS}.  
As for charged asymptotically anti-de Sitter (AdS) black holes, 
they have found no evidence for instability 
in $D=5,...,11$\cite{Konoplya&Zhidenko08ChargedAdS}.

\subsection{Lovelock black holes}  
\label{subsec:Lovelock}

As we have seen in section~\ref{sec:Lovelock}, there are known 
exact solutions of static black holes in Lovelock theory, 
and the master equations for all types of perturbations of static vacuum
black holes in general Lovelock theory have recently been derived
by Takahashi and Soda\cite{Takahashi.T&Soda2010A}. 
Using the master equations, they have found that an asymptotically flat,
static Lovelock black hole with small mass is unstable in arbitrary higher dimensions; 
it is unstable with respect to tensor type perturbations in 
even-dimensions\cite{Takahashi&Soda09a,Takahashi&Soda09b} and 
with respect to scalar type perturbations 
in odd-dimensions\cite{Takahashi&Soda10b}. 
The stability under vector type perturbations in all dimensions has also
been shown\cite{Takahashi&Soda10b}  
by applying the $S$-deformation technique.   

In fact, such an instability against tensor and scalar type perturbations, 
as well as the stability under vector type perturbations, have already been 
indicated by earlier work\cite{Gleiser.R&Dotti05,Dotti.G&Gleiser05b,Dotti.G&Gleiser05a,BDG07} 
performed within the framework of second-order Lovelock theory, 
often called the Einstein-Gauss-Bonnet theory. For such a restricted
class of Lovelock theory--though most generic in $D=5,6$, the master equations 
for metric perturbations have previously been derived by 
Dotti and Gleiser\cite{Dotti.G&Gleiser05b,Gleiser.R&Dotti05}.  
A numerical analysis of the (in)stability of static black holes in 
Einstein-Gauss-Bonnet theory in dimensions $D=5, ..., 11$ has been performed 
in Ref.~\citen{Konoplya&Zhidenko08}. 

%
It is interesting to note that the instability found in small Lovelock black holes 
is typically stronger in short distance scales rather than long distance 
scale/low multipoles as one may expect. 
For example, for tensor type perturbations, there appears the eigenvalue 
of tensor harmonics on the horizon manifold as an overall factor 
in the effective potential term of Eq.~\eqref{EVofA:S}, 
and it is therefore always possible to make the potential term dominant 
by taking a sufficiently large eigenvalue of tensor harmonics. 
This implies that if the effective potential term can be negative, 
the right-hand side of Eq.~\eqref{EVofA:S}, as a whole, can also be negative. 
This is shown to be the case when the mass is sufficiently small. 
A similar argument also applies to the case of scalar type perturbations 
at higher multipole moments\cite{Takahashi&Soda10b}. 

\section{Summary and Discussions}  
\label{sec:SummaryDiscussions}

We have reviewed a gauge invariant formalism for gravitational and 
electromagnetic perturbations of static charged black holes with 
cosmological constant in higher dimensions and, as an application, 
the stability analysis using the master equations derived in the 
developed formalism. 
In section~\ref{sec:background:decomposition:harmonics}, we have started 
with considering a fairly generic class of background spacetimes defined by 
a warped product of an $m$-dimensional spacetime and an $n$-dimensional 
internal space, where the latter corresponds to the horizon cross-section 
manifold. We have explained how to decompose tensor fields in the background 
spacetime from the viewpoint of the internal space, and have seen 
that second-rank symmetric tensor fields or metric perturbations 
are decomposed into tensor, vector and scalar-types. 
We have then constructed manifestly gauge invariant variables for each type of 
perturbations. 
After that in section~\ref{sec:Harmonictensors}, 
we have introduced harmonic tensors on the internal space to expand 
the gauge invariant variables in terms of them. We have presented 
several theorems concerning basic properties of harmonic tensors, and 
also given explicit expressions of the harmonic tensors in terms of 
homogeneous coordinates. 
In subsequent sections~\ref{sec:TensorPerturbation}
-\ref{sec:ScalarPerturbation}, we have briefly described how to
reduce the perturbed Einstein and Maxwell equations written in terms
of the gauge-invariant variables to a set of decoupled master equations
for a single scalar variable on the $2$-dimensional background spacetime
for the vector and scalar type perturbations. For the tensor type
perturbations, there is no electromagnetic perturbation mode and
the reduction can immediately be done to obtain the master equation
on the generic $m$-dimensional spacetime.  
In the black hole background case, our master equations generalize to 
higher dimensions in a manifestly gauge-invariant manner the well-known 
Regge-Wheeler-Zerilli equations for Schwarzschild black holes,  
and Moncrief's equations\cite{Moncrief74a,Moncrief74b} 
for Reissner-Nordstr\"{o}m black holes, as well as the master equations
given by Cardoso and Lemos\cite{Cardoso.V&Lemos2001} that include cosmological 
constant in $4$-dimensions.   
We have seen that by taking Fourier decomposition with respect to the time coordinate, 
each of the master equations is expressed in the form of a one-dimensional 
self-adjoint ODE. 
Therefore, as guaranteed by spectral theory of self-adjoint operators, 
the master equations can govern all possible perturbations that are 
normalizable with respect to the standard inner product. 
This is in particular important when we address the stability problem of a 
given solution, as the stability proof should be a statement concerning 
all physically acceptable (normalizable) perturbations. 
This is in contrast to the rotating black hole case, 
for which we have Teukolsky's equations in $4$-dimensions and 
a similar set of master
equations\cite{KLR06,Murata&Soda08CQG,Murata&Soda08PTP}
for some special cases of higher dimensional Myers-Perry black holes, but 
those master equations do not reduce to a form of the self-adjoint 
eigenvalue problem, except for some special 
modes [c.f., Ref.~\citen{Kodama.H&Konoplya&Zhidenko09}].   
In section~\ref{sec:Lovelock}, we have also briefly reviewed a similar type of 
master equations for gravitational perturbations of static black holes 
in generic Lovelock theory, derived recently\cite{Takahashi.T&Soda2010A}.

As an immediate and one of the most important applications of the master 
equations, we have examined the stability of higher dimensional static black 
holes with charge and cosmological constant in 
section~\ref{sec:stabilityanalysis}.  
The task is to show the positivity of the self-adjoint operator 
appeared as the spatial derivative part of the master equations. 
For higher dimensional black holes, the potential term of the relevant 
self-adjoint operator is in general not positive definite, 
in contrast to the $4$-dimensional case, and the stability is therefore 
not taken for granted in higher dimensions. 
To deal with such situations, we have developed the $S$-deformation technique, 
in which our task is to find some suitable function, $S$, that makes 
the effective potential deformed by $S$ positive definite.  
Having applied this technique we have shown that a large variety of 
static black holes in higher dimensional general relativity are stable 
under gravitational as well as electromagnetic perturbations 
as summarized in Table~\ref{tbl:stability}.
This technique has also applied to the (in)stability analysis in general
Lovelock theory, as has just been discussed in
section~\ref{subsec:Lovelock} above.
A similar type of stability analysis using the $S$-deformation, 
restricted to tensor-type perturbations, has also been performed 
for static black holes 
in higher derivative stringy gravity\cite{Moura.F&Schiappa07}.  
Thus, the $S$-deformation has turned out to be a powerful tool to
address stability problem.

However, as also indicated in Table~\ref{tbl:stability}, 
the stability analysis of higher dimensional black holes has not been 
completed yet even within the context of general relativity.     
For some cases, in particular, when electric charge and cosmological 
constant are involved, we have not yet been able to draw definite 
conclusions for the stability problem with respect to scalar type 
perturbations. This is in part because there does not seem to be 
a systematic method to find such a desirable function $S$ that could apply 
to generic cases. 
In particular, for the charged black hole case, the potential for 
scalar-type perturbation admits a negative ditch in the 
immediate vicinity of the horizon, which appears to be difficult 
to remove by the $S$-deformation.
In connection to this, it would be interesting to note that the existence of 
such a negative ditch in the potential may have a significant influence
on the frequencies of the quasinormal modes and the graybody
factor for the Hawking process, even if these black holes are
found to be stable. 
Also, given a choice of $S$, the stability would still depend upon 
the range of the parameters characterising the black hole solution 
as well as upon the eigenvalues of the harmonic functions. 
Therefore our analysis using the $S$-deformation needs to go on 
a case-by-case basis.  
A numerical analysis\cite{Konoplya&Zhidenko08ChargedAdS} has 
found no indication of instability for charged AdS black holes 
in $D=5,..., 11$. 
Therefore, at least for charged AdS black holes 
it may still be possible to analytically prove its stability by using the $S$-deformation or 
other analytic methods. 

As other applications than the stability issue, 
the master equations can be used in numerical studies of 
black hole quasinormal modes in higher dimensions\cite{Cardoso.V&&2003A,Cardoso.V&Dias&Lemos2003,Cardoso.V&Lemos2001,Cardoso.V&Lemos2003A,Cardoso.V&Konoplya&Lemos2003A,Konoplya.R2003A}. 
It would also be interesting to consider stationary perturbations 
that could describe deformation of the event horizon, as considered in 
Ref.~\citen{Konoplya&Zhidenko09ChargeddS} for unstable 
charged de Sitter black holes.  
If one finds no stationary perturbation that is regular everywhere on and 
outside the event horizon, then it would support the uniqueness 
property\cite{Hwang.S1998,Gibbons.G&Ida&Shiromizu2002a} 
of the given background black hole solution, as analysed in Paper~II 
for the vacuum black hole case, as well as in Ref.~\citen{Kodama.H2004} 
for more general cases. This may be interesting in particular 
in asymptotically AdS black hole case to find 
deformed horizon solutions [see Ref.~\citen{Tomimatsu.A2004} 
for such an analysis in $4$-dimensions]. 
Also, we note that depending on the type of perturbations and 
the dimensionality, asymptotically AdS black holes admit a large class 
of boundary conditions at conformal infinity, 
other than the Dirichlet conditions considered in 
section~\ref{sec:stabilityanalysis}. 
As has been considered in the context of gauge/gravity correspondence 
and often examined within AdS gravity coupled to a scalar 
field\cite{FHR10,HorowitzHubeny}, 
it would be interesting to clarify whether different choice of boundary 
conditions leads to different consequences for the stability problem.

\section*{Acknowledgements}
AI is supported by the  JSPS Grant-in-Aid for Scientific Research 
(A)No. 22244030 and (C)No. 22540299. 
HK is supported by  the JSPS Grant-in-Aid for Scientific Research 
(A)No. 22244030 and the the MEXT Grant-in-Aid for Scientific Research on
Innovative Areas No. 21111006.


\end{document}